\documentclass[authorversion, nonacm, 11pt, screen]{acmart}

\usepackage[utf8]{inputenc}

\usepackage{xspace}

\usepackage[toc,page]{appendix}

\usepackage{cleveref}
\crefname{section}{\S}{\S\S}
\Crefname{section}{\S}{\S\S}
\crefname{figure}{Figure}{Figures}
\Crefname{figure}{Figure}{Figures}
\crefformat{section}{\S#2#1#3}

\usepackage[normalem]{ulem}

\hypersetup{allcolors=blue}

\newcommand{\name}{AT2\xspace}
\newcommand{\namesh}{AT2$_{SM}$\xspace}
\newcommand{\namemp}{AT2$_{MP}$\xspace}
\newcommand{\namedet}{AT2$_{\mathsf{D}}$\xspace}
\newcommand{\nameprob}{AT2$_{\mathsf{P}}$\xspace}

\newcommand{\operation}[2]{$\textit{#1}#2$}

\newcommand{\myparagraph}[1]{\vspace{3pt}\noindent \textbf{#1}}
\def\Nat{\mathbb{N}}
\def\to{\rightarrow}

\def\transfer{\textit{transfer}}
\def\read{\textit{read}}

\def\true{\textit{true}}
\def\false{\textit{false}}

\def\lf{\tiny}

\def\nnll{\refstepcounter{linenumber}{\lf\thelinenumber}}
\newcommand{\commentline}[1]{\hspace{1cm}\{ \textit{#1} \}}
\newcommand{\commentsimple}[1]{\{ \textit{#1} \}}
\newcounter{linenumber}

\newcommand{\A}{\mathcal A}

\newcommand{\M}{\mathcal M}
\newcommand{\V}{\mathcal V}

\newcommand{\kcons}{$k$-consensus}

\newcommand{\kasset}{$k$-shared asset transfer}
\newcommand{\ignore}[1]{}

\newcommand{\cmt}[2]{}

\usepackage{xifthen}
\usepackage{color}
\usepackage{tabto}
\usepackage[toc,page]{appendix}
\usepackage{cancel}
\usepackage{comment}

\usepackage{graphicx}
\usepackage{tikz}
\usepackage{pgfplots}
    \definecolor{palette-0}{RGB}{27,158,119}
    \definecolor{palette-1}{RGB}{217,95,2}
    \definecolor{palette-2}{RGB}{117,112,179}
    \definecolor{palette-3}{RGB}{231,41,138}
    \definecolor{palette-4}{RGB}{102,166,30}
    \definecolor{palette-5}{RGB}{230,171,2}
    \definecolor{palette-6}{RGB}{166,118,29}
    \definecolor{palette-7}{RGB}{102,102,102}

\usepackage{amsthm}
\usepackage{amsmath}
\usepackage{amsfonts}
\usepackage{amssymb}

\usepackage{algorithmicx}
\usepackage{algorithm}
\usepackage{algpseudocode}

\newcommand{\broadcast}{probabilistic broadcast} 
\newcommand{\Broadcast}{Probabilistic broadcast}
\newcommand{\broadcastabstraction}{ProbabilisticBroadcast}
\newcommand{\broadcastinstance}{pb}

\newcommand{\singleshot}{probabilistic consistent broadcast}
\newcommand{\Singleshot}{Probabilistic consistent broadcast}
\newcommand{\singleshotabstraction}{ProbabilisticConsistentBroadcast}
\newcommand{\singleshotinstance}{pcb}

\newcommand{\multishot}{probabilistic secure broadcast}
\newcommand{\Multishot}{Probabilistic secure broadcast}
\newcommand{\multishotabstraction}{ProbabilisticSecureBroadcast}
\newcommand{\multishotinstance}{psb}

\newcommand{\erg}{\texttt{Erdös}-\texttt{Rényi Gossip}}
\newcommand{\pde}{\texttt{Probabilistic Double}-\texttt{Echo}}
\newcommand{\tfayto}{\texttt{Sequenced Probabilistic Double}-\texttt{Echo}}

\newcommand{\pebbling}{\texttt{Threshold Contagion}}

\newcommand{\event}[3]{\langle #1.\textrm{#2} \mid #3\rangle}
\newcommand{\sevent}[2]{\langle #1.\textrm{#2} \rangle}

\algnewcommand\Instance[2]{\State #1, \textbf{instance} #2}
\algnewcommand\Trigger[3]{\State \textbf{trigger} $\event{#1}{#2}{#3}$}
\algnewcommand\sTrigger[2]{\State \textbf{trigger} $\event{#1}{#2}$}

\algsetblockdefx[Implements]{Implements}{EndImplements}{}{}{\textbf{Implements:}}{}
\algsetblockdefx[Uses]{Uses}{EndUses}{}{}{\textbf{Uses:}}{}
\algsetblockdefx[Parameters]{Parameters}{EndParameters}{}{}{\textbf{Parameters:}}{}

\algsetblockdefx[Upon]{Upon}{EndUpon}{}{}[3]{\textbf{upon event} $\event{#1}{#2}{#3}$ \textbf{do}}{}
\algsetblockdefx[Upon]{sUpon}{EndsUpon}{}{}[2]{\textbf{upon event} $\sevent{#1}{#2}$ \textbf{do}}{}
\algsetblockdefx[UponExists]{UponExists}{EndUponExists}{}{}[2]{\textbf{upon exists} $#1$ \textbf{such that} $#2$ \textbf{do}}{}
\algsetblockdefx[Upon]{cUpon}{EndcUpon}{}{}[1]{\textbf{upon} $#1$ \textbf{do}}

\algsetblockdefx[Procedure]{Procedure}{EndProcedure}{}{}[2]{\textbf{procedure} \emph{#1}($#2$) \textbf{is}}{}

\algsetblockdefx[If]{If}{EndIf}{}{}[1]{\textbf{if} $#1$ \textbf{then}}{\textbf{end if}}
\algsetblockdefx[Ntimes]{Ntimes}{EndNtimes}{}{}[1]{\textbf{for} $#1$ \textbf{times do}}{\textbf{end for}}
\algsetblockdefx[ForAll]{ForAll}{EndForAll}{}{}[2]{\textbf{for all} $#1 \in #2$ \textbf{do}}{\textbf{end for}}

\theoremstyle{definition}
    \newtheorem{definition}{Definition}

\theoremstyle{plain}
    \newtheorem{lemma}{Lemma}
    \newtheorem{theorem}{Theorem}

\newcommand{\rp}[1]{{\left(#1\right)}}

\newcommand{\abs}[1]{{\left|#1\right|}}

\authorsaddresses{}

\begin{document}

\title{\name: Asynchronous Trustworthy Transfers}

\author{\vspace{-.15cm}Rachid Guerraoui}
\affiliation{
  \institution{EPFL}
}

\author{\vspace{-.15cm}Petr Kuznetsov}
\affiliation{
  \institution{LTCI, T\'el\'ecom ParisTech, University Paris-Saclay}
}

\author{\vspace{-.15cm}Matteo Monti}
\affiliation{
 \institution{EPFL}
}
\author{\vspace{-.15cm}Matej Pavlovic}
\affiliation{
  \institution{EPFL}
}
\author{\vspace{-.15cm}Dragos-Adrian Seredinschi}
\affiliation{
  \institution{EPFL}
}

\renewcommand{\contentsname}{}
\renewcommand\shortauthors{GKMPS'18}

\maketitle


{\large\textbf{ABSTRACT }}
Many blockchain-based protocols, such as Bitcoin, implement a decentralized asset transfer system.
As clearly stated in the original paper by Nakamoto, the crux of this problem lies in prohibiting any participant from engaging in \emph{double-spending}.
There seems to be a common belief that consensus is necessary for solving the double-spending problem.
Indeed, whether it is for a permissionless or a permissioned environment, the typical solution uses consensus to build a totally ordered ledger of submitted transfers.

In this paper we show that this common belief is false:
consensus is not needed to implement a decentralized asset transfer system.
We do so by introducing \name (Asynchronous Trustworthy Transfers), a class of consensusless algorithms.

To show formally that consensus is unnecessary for asset transfers, we first consider this problem in the shared-memory context.
We introduce \namesh, a wait-free algorithm that asynchronously implements asset transfer in the \emph{read-write} shared-memory model.
In other words, we show that the \emph{consensus number} of an asset-transfer object is \emph{one}.

In the message passing model with Byzantine faults, we introduce a generic \emph{asynchronous} algorithm called \namemp and discuss two instantiations of this solution.
First, \namedet ensures \emph{deterministic} guarantees and consequently targets a small scale deployment (tens to hundreds of nodes), typically for a private, i.e, permissioned, environment.
Second, \nameprob provides \emph{probabilistic} guarantees and scales well to a very large system size (tens of thousands of nodes), ensuring logarithmic latency and communication complexity.
Instead of consensus, we construct \namedet and \nameprob on top of a broadcast primitive with causal ordering guarantees offering deterministic and probabilistic properties, respectively.

Whether for the deterministic or probabilistic model, our \name algorithms are both simpler and faster than solutions based on consensus.
In systems of up to $100$ replicas, regardless of system size, \namedet outperforms consensus-based solutions offering a throughput improvement ranging from $1.5x$ to $6x$, while achieving a decrease in latency of up to $2x$.
(Not shown in this version of the document.)
\nameprob obtains \emph{sub-second} transfer execution on a global scale deployment of thousands of nodes.



\newpage
\setcounter{tocdepth}{1}
\tableofcontents

\newpage

\section{Introduction}
\label{sec:intro}

In $2008$, Satoshi Nakamoto introduced the Bitcoin protocol, implementing an electronic asset transfer system (often called a cryptocurrency) without any central authority~\cite{nakamotobitcoin}.
Since then, many alternatives to Bitcoin came to prominence, designed for either the \emph{permissionless} (public) or \emph{permissioned} (private) setting.
These include major cryptocurrencies such as Ethereum~\cite{ethereum} or Ripple~\cite{rapoport2014ripple}, as well as systems sparked from research or industry efforts such as Bitcoin-NG~\cite{eyal16bitcoinng}, Algorand~\cite{gilad2017algorand}, ByzCoin~\cite{kogi16byzcoin}, Stellar~\cite{mazieres2015stellar}, Hyperledger~\cite{hyperledger}, Corda~\cite{he16corda}, or Solida~\cite{abr16solida}.
Each of these alternatives brings novel approaches to implementing decentralized transfers, and may offer a more general interface in the form of smart contracts~\cite{sz97smartcontr}.
They improve over Bitcoin in various aspects, such as performance, scalability, energy-efficiency, or security.

A common theme in these protocols, whether they are for transfers~\cite{kok18omniledger} or smart contracts~\cite{ethereum},
is that they seek to implement a \emph{blockchain}.
This is a distributed ledger where all the transfers in the system are totally ordered.
Achieving total order among multiple inputs (e.g., transfers) is fundamentally a hard task, equivalent to solving \emph{consensus}~\cite{FLP85,HT93}.

Consensus is a central problem in distributed computing, known for its notorious difficulty.
Consensus has no deterministic solution in asynchronous systems
if just a single participant can fail~\cite{FLP85}.
Algorithms for solving consensus are tricky to implement correctly~\cite{abrah17revisiting,cac17blwild,clement09making}, and they face tough trade-offs between performance, security, and energy-efficiency~\cite{antoni18smr,ber89optimal,gue18blockchain,vuko15quest}.

As stated in the original paper by Nakamoto, the main problem of a decentralized cryptocurrency is preventing a malicious participant from spending the same money more than once~\cite{nakamotobitcoin}.
This is known as the \emph{double-spending} attack.
Bitcoin and follow-up systems typically assume that total order---and thus consensus---is vital to preventing double-spending~\cite{gar15backbone}.
Indeed, there seems to be a common belief that solving consensus is necessary for implementing a decentralized asset transfer system~\cite{BonneauMCNKF15,gue18blockchain,kar18vegvisir,nakamotobitcoin}.

Our main contribution in this paper is showing that this common belief is false.
We show that total order is not required to avoid double-spending in decentralized transfer systems.
We do so by introducing \name (Asynchronous Trustworthy Transfers), a class of consensusless algorithms.

As a starting point, we consider the shared memory model with benign failures.
In this model, we give a precise definition of a transfer system as a sequential object type.
It is for pedagogical purposes that we start from this model, as it allows us to study the relation between the transfer object and consensus, and later extend our result to the message passing model.
%
We introduce \namesh, an asset transfer algorithm that has consensus number \emph{one}~\cite{Her91}.
In other words, decentralized asset transfer does not need consensus in its implementation, and we can avoid maintaining a total order across transfers.


To get an intuition why total order is not necessary, consider a set of users
who transfer money between their accounts in a decentralized manner.
For simplicity, assume the full replication model, where every user maintains a copy of the state of every account.
We observe that most operations in a transfer system \emph{commute}, i.e., different users can apply the operations in arbitrary order, resulting in the same final state.
For instance, a transfer $T_1$ from Alice to Bob commutes with a transfer $T_2$ from Carol to Drake.
This is because the two transfers involve different accounts.
In the absence of other transfers, $T_1$ and $T_2$ can be applied in different orders by different users, without affecting correctness.

Consider now a more interesting case when two transfers involve the same account.
For example, let us throw into the mix a transfer $T_0$ from Alice to Carol.
Assume that Alice issues $T_0$ before she issues $T_1$.
Note that $T_0$ and $T_1$ do not commute, because they involve the same account---that of Alice---and
it is possible that she cannot fulfill both $T_0$ and $T_1$ (due to insufficient balance).
We say that $T_1$ depends on $T_0$, and so $T_0$ should be applied before $T_1$.

Furthermore, suppose that Carol does not have enough money in her account to fulfill $T_2$ before she receives transfer $T_0$ from Alice.
In this case, transfer $T_2$ depends on $T_0$.
Thus, all users should apply $T_0$ before applying $T_2$, while transfers $T_1$ and $T_2$ still commute.
The partial ordering among these three transfers is in fact given by a \emph{causality} relationship.

In~\Cref{fig:scenario-intro} we show the scenario with these three transfers and the causality relation between them.
Intuitively, the general ordering constraint we seek to enforce is that every outgoing transfer for an account causally depends on all preceding transfers involving that---and only that---account.

\begin{figure}[h!]
\includegraphics[width=0.4\columnwidth]{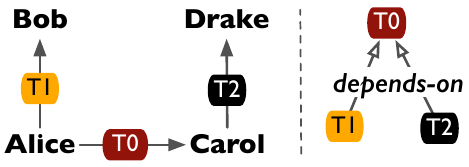}
\caption{A simple scenario illustrating, on the left side, three transfers among four participants.
On the right side, we depict the dependencies among these three transfers.
The dependencies are established by a causality-based ordering relationship.}
\label{fig:scenario-intro}
\end{figure}


We then extend our result to the model of Byzantine fault-prone processes that communicate via passing messages.
In this model, we show how to sidestep consensus by presenting \namemp, a generic algorithm that implements decentralized asset transfer atop a variant of \emph{secure broadcast}.
We describe two variants of this generic solution.
The first,  called {\namedet},  has \emph{deterministic} guarantees that targets a smaller scale deployment (tens to hundreds of nodes) typically for a private environment.
The second, called {\nameprob}, ensures probabilistic guarantees and scales well to very large system sizes, typically for a public setting.

It is well-known that the main bottleneck in blockchain-based systems is their consensus module~\cite{he16corda,sou18byzantine,vuko15quest}.
Numerous solutions have emerged to alleviate this problem~\cite{gol18sbft,he16corda}.
Typical techniques seek to employ a form of sharding~\cite{kok18omniledger}, for instance, or to use a committee-based optimization~\cite{eyal16bitcoinng,gilad2017algorand}.
With \name, we circumvent this main bottleneck, yielding new solutions that bypass consensus altogether.



Whether for the deterministic or probabilistic model, our \namemp\ algorithms are both simpler and faster than solutions based on consensus.
In systems of up to $100$ replicas, regardless of system size, \namedet\ outperforms consensus-based solutions offering a throughput improvement ranging from $1.5x$ to $6x$, while achieving a decrease in latency of up to $2x$.
\nameprob\ obtains \emph{sub-second} transfer execution on a global scale deployment of thousands of nodes.
\cmt{RG}{I believe it is crucial to talk about number of lines of code (order of magnitude less) and say that this is crucial for proving and testing the system. 
You might also say that our protocols are all built on solid foundations and proven correct}
\cmt{YAP}{not shown in this version... it would be great to mention the complexity instead!}

The rest of this paper is organized as follows.
We first define the asset transfer object type in the crash-stop shared memory model and show that it has consensus number one for accounts with a single owner~(\cref{sec:result}).
We then move to the message passing model with Byzantine failures and present \namemp, a generic algorithm for implementing distributed asset transfer on top of a secure broadcast primitive~(\cref{sec:algo}).
Next, we focus on one deterministic and one probabilistic version of the secure broadcast primitive, yielding, respectively, \namedet~(\cref{sec:algo-det}) and \nameprob~(\cref{sec:algo-prob}), our deterministic and probabilistic flavors of \name\ for the Byzantine message passing model.
At the end, we revisit the crash-stop shared memory model to prove that a general asset transfer object has consensus number $k$ if an account is shared by $k$ (but not more) processes (\cref{sec:k-consensus}).
Finally, we discuss related work (\cref{sec:related-work}) and conclude (\cref{sec:conclusions}).
Insights behind the analyses of our algorithms can be found in the appendices.



\section{Consensusless Asset Transfer With \namesh}
\label{sec:result}

In this section we formally define the asset transfer problem and discuss its consensus number.
We begin with presenting the shared memory model we use in this section~(\cref{sec:def-shm-model}) and then precisely define the problem of \textsf{asset-transfer} as a sequential object type~(\cref{sec:def-sequential}).

Intuitively, an \textsf{asset-transfer} object consists of accounts whose balances can be read by all processes.
Processes are also allowed to transfer assets between accounts,
where each account is associated with a subset of processes (owners) that are allowed to issue transfers debiting this account.

Finally, we prove that an object of this type where each account has a single owner has consensus number \emph{one}.
We do so by proposing a wait-free implementation of this object called \namesh ~(\Cref{sec:avoiding-consensus}).
The interested reader is referred to \cref{sec:k-consensus} that generalizes this result, proving that an object containing an account with $k$ owners has consensus number $k$.

\subsection{Shared Memory Model Definitions}
\label{sec:def-shm-model}

\myparagraph{Processes.}
As our basic system model, we assume a set $\Pi$ of $N$ asynchronous processes
that communicate by invoking atomic operations on shared memory objects.
Processes are sequential---we assume that a process never invokes a new operation before obtaining a response from a previous one.

\myparagraph{Object types.} 
A sequential object type is defined as a tuple
$T=(Q,q_0,O,R,\Delta)$, where $Q$ is a set of states, $q_0\in Q$ is an
initial state, $O$ is a set of operations, $R$ is a set of responses and
$\Delta\subseteq Q\times\Pi\times O \times Q \times R$ is a relation
that associates a state, a process identifier and an operation to a set of
possible new states and corresponding responses.
Here we assume that $\Delta$ is total on the first three elements,
i.e., for each state $q\in Q$, each process $p\in \Pi$, and each operation $o\in O$, some
transition to a new state is defined, i.e., $\exists q'\in Q,\; r\in R$: $(q,p,o,q',r)\in \Delta$.

A \emph{history} is a sequence of invocations and responses, each
invocation or response associated with a process identifier.
A
\emph{sequential history} is a history that starts with an invocation and in
which every invocation is immediately followed with a response
associated with the same process.    
A sequential history $(j_1,o_1),(j_1,r_1),(j_2,o_2),(j_2,r_2),\ldots$, where $\forall i\ge
1, j_i\in\Pi,\; o_i\in O,\; r_i\in R$, is \emph{legal} with respect to type 
$T=(Q,q_0,O,R,\Delta)$ if there exists a
sequence $q_1,q_2,\ldots$ of states in $Q$ such that
$\forall i\geq 1$, $(q_{i-1},j_i,o_i,q_i,r_i)\in \Delta$. 

\myparagraph{Implementations.} 
An \emph{implementation} of an object type $T$ is a distributed algorithm that,
for each process and invoked operation, prescribes the actions that the process needs to
take to perform it.
An \emph{execution} of an implementation is a sequence of
\emph{events}: invocations and responses of operations, \emph{send}
and \emph{receive} events, or atomic accesses to shared abstractions. The sequence of events at every process
must respect the algorithm assigned to it. 

\myparagraph{Failures.}
Processes are subject to \emph{crash} failures. A process may halt prematurely, in which case we say that the process is \emph{crashed}.
A process is called \emph{faulty} if it crashes during the execution. 
A process is \emph{correct} if it is not faulty.
All algorithms we present in the shared memory model are \emph{wait-free}---every correct process eventually returns from each operation it invokes,
regardless of an arbitrary number of other processes crashing.

\myparagraph{Linearizability and sequential consistency.}
For each pattern of operation invocations, the execution produces a
\emph{history}, i.e., the sequence of distinct invocations and responses,
labelled with process identifiers and unique sequence numbers.

A projection of a history $H$ to process $p$, denoted $H|p$ is the
subsequence of elements of $H$ labelled with $p$. 
An invocation $o$ by a process $p$ is \emph{incomplete} in $H$ if it is not followed by a
response in $H|p$.   
A history is \emph{complete} if it has no incomplete invocations.
A \emph{completion} of $H$ is a history $\bar H$ that is identical to
$H$ except that every incomplete invocation in $H$ is either removed or
\emph{completed} by inserting a matching response somewhere after it.  

A \emph{sequentially consistent} implementation of $T$ ensures that for every
history $H$ it produces, there exists a completion $\bar H$ and a
legal sequential history $S$ such that for all processes $p$, $\bar H|p=S|p$.

A \emph{linearizable}  implementation, additionally, preserves the
real-time order between operations.
Formally, an invocation $o_1,r_1$ \emph{precedes} an
invocation $o_2$ in $H$, denoted  $o_1\prec_H o_2$, if $o_1$ is
complete and the corresponding response $r_1$ precedes $o_2$ in $H$.
Note that $\prec_H$ stipulates a partial order on invocations in $H$.
A linearizable implementation of $T$ ensures that for every history $H$ it produces, there exists a completion $\bar H$ and a
legal sequential history $S$ such that (1)~for all processes $p$, $\bar H|p=S|p$ and (2)~$\prec_H\subseteq\prec_S$.

A (sequentially consistent or linearizable) implementation is
\emph{$t$-resilient} if, under the assumption that at most $t$
processes crash, it ensures that every invocation performed by a
correct process is eventually followed by a response.
In the special case when $t=n-1$, we say that the implementation is
\emph{wait-free}. 

\cmt{YAP}{
This definition is restricted to 1 input account, bitcoin and other currencies support transactions with multiple input accounts (multisig transactions).
For a constant number of input accounts everything in the analysis works out, without increasing the consensus number...
(multisig transactions are used in many smart contracts)
We could mention this in the extension section.
}
\cmt{<--MP}{
This would be a little bit trickier than in Bitcoin, as each partial transfer would need to be ordered separately with respect to each source account.
But I believe it is possible to implement some equivalent ``all or nothing'' semantics \name\ as well.
}

\subsection{Asset transfer type}
\label{sec:def-sequential}

Let $\A$ be a set of \emph{accounts} and $\mu: \A \rightarrow
2^{\Pi}$ be an ``owner'' map that associates each account with a set
of processes that are, intuitively, allowed to debit the account.  
The \textsf{asset-transfer} object type
associated with $\A$ and $\mu$
is then defined as a tuple $(Q,q_0,O,R,\Delta)$, where:

\begin{itemize}

\item The set of states $Q$ is the set of all possible maps
  $q:\;\A\to\Nat$. Intuitively, each state of the object assigns
  each account its \emph{balance}.   

\item The initialization map  $q_0:\;\A\to\Nat$ assigns the initial
  balance to each account.

\item Operations and responses of  the type are defined as
  $O=\{\transfer(a,b,x):\; a,b\in\A,\,x\in\Nat\}\cup\{\textit{read}(a):\;a\in\A\}$
  and $R=\{\true,\false\}\cup\Nat$.
  \cmt{YAP}{extend to multi account transactions here or in Sec. 6}

 \item For a state $q\in Q$, a proces $p\in\Pi$, an operation $o\in O$, a response
   $r\in R$ and a new state $q'\in Q$, the tuple $(q,p,o,q',r)\in\Delta$
   if and only if one of the following conditions is satisfied:

  \begin{itemize}

     \item $o=\transfer(a,b,x) \wedge p\in\mu(a)$ $\wedge$
       $q(a)\geq x$ $\wedge$ $q'(a)=q(a)-x$ $\wedge$ $q'(b)=q(b)+x \wedge \forall c \in{\A}\setminus\{a, b\}: q'(c) = q(c)$ (all other
       accounts unchanged) $\wedge$ $r=\true$;

     \item $o=\transfer(a,b,x)$ $\wedge$  ($p\notin \mu(a)$
       $\vee$ $q(a)< x$) $\wedge$ $q'=q$ $\wedge$ $r=\false$;
       
     \item $o=\textit{read}(a)$ $\wedge$ $q=q'$ $\wedge$ $r=q(a)$. 

   \end{itemize} 

  \end{itemize}

In other words, operation $\transfer(a,b,x)$ invoked by
process $p$ \emph{succeeds} if and only if $p$ is the owner of the
source account $a$ and 
account $a$  has enough balance, and if it does, $x$ is transferred
from $a$ to the destination account $b$.
A $\transfer(a,b,x)$ operation is called \emph{outgoing} for
$a$ and \emph{incoming} for $b$; respectively, the $x$ units are called   \emph{outgoing} for
$a$ and \emph{incoming} for $b$. 
A transfer is \emph{successful} if its corresponding response is \emph{true} and \emph{failed} if its corresponding response is \emph{false}.
Operation $\textit{read}(a)$ simply returns the balance of $a$ and
leaves the account balances untouched.   

\subsection{Asset Transfer Is Easier than Consensus}
\label{sec:avoiding-consensus}

In this section, we discuss the ``synchronization power'' of the \textsf{asset-transfer} type in what we believe to be the most typical use case:
each account being associated with a single owner process.
We show that such an \textsf{asset-transfer} object type can be implemented in a \emph{wait-free} manner using only read-write registers.
Thus, the type has consensus number $1$.
In \cref{sec:k-consensus} we generalize our result and show that a ``$k$-shared'' \textsf{asset-transfer} object (where up to $k$ processes may share an account) has consensus number $k$.

For now, consider an \textsf{asset-transfer} object associated with a set of accounts $\A$ and
an ownership map $\mu$ such that $\forall a\in\A$, $|\mu(a)|\leq 1$.
\cmt{YAP}{when does equality make sense?}
\cmt{<--MP}{In fact, only equality makes sense here, in which case an account is only debited by one owner. Zero makes less sense, as this would correspond to an account that no process can withdraw from. But technically this is possible.}
We now present \namesh, our algorithm that implements this object in the read-write shared-memory model.  

Our implementation is described in Figure~\ref{fig:waitfree}. 
The $N$ processes share an atomic snapshot object~\cite{AADGMS93} of
size $N$.
Every process $p$ is associated with a distinct location in the
atomic snapshot object
storing the set of all successful {\transfer} operations
executed by $p$ so far.
Since each account is owned by at most one process, all outgoing transfers for an account appear in a single location of the atomic snapshot (associated with the owner process).

Recall that the atomic snapshot (AS) memory is represented as a vector of $N$ 
shared variables
that can be accessed with two atomic operations: \emph{update} and 
\emph{snapshot}. An \emph{update} operation  
modifies the value at a given position of the vector and a
\emph{snapshot} returns the state of the whole vector.
We implement the \textit{read} and \textit{transfer} operations as follows.

\begin{itemize}
\item
To read the balance of an account $a$, the process simply takes a snapshot $S$ and
returns the initial balance plus the sum of incoming amounts minus the sum of all outgoing
amounts. We  denote this number by $\textit{balance}(a,S)$.
As we argue below, the result is guaranteed to be non-negative,
i.e., the operation is correct with respect to the type specification.
\item
To perform $\transfer(a,b,x)$, a process $p$, the owner of $a$, takes a snapshot $S$ and computes $\textit{balance}(a,S)$. 
If the amount to be transferred does not exceed $\textit{balance}(a,S)$,
we add the transfer operation to the set of $p$'s operations in the snapshot object via an \textit{update} operation and return $\true$.
Otherwise, the operation returns $\false$.
\end{itemize}

\begin{figure}[tbp]
\hrule \vspace{1mm}
 {\small
\begin{tabbing}
 bbb\=bb\=bb\=bb\=bb\=bb\=bb\=bb \=  \kill
Shared variables: \\
\> $AS$, atomic snapshot, initially $\{\bot\}^N$\\
\\
Local variables: \\
\> $\textit{ops}_p\subseteq \A\times\A\times\Nat$, initially $\emptyset$\\
\\
Upon \transfer$(a,b,x)$\\
\nnll\label{line1:tsf:snapshot}\> $S := AS.\textit{snapshot}()$\\ 
\nnll\> \textbf{if} $p \notin \mu(a) \vee \textit{balance}(a,S)< x$ \textbf{then}\\
\nnll\label{line1:response-false}\>\> \textbf{return} {\false}\\ 
\nnll\> $\textit{ops}_p := \textit{ops}_p \cup \{(a,b,x)\}$\\
\nnll\label{line1:tsf:update}\> $AS.\textit{update}(\textit{ops}_p)$\\
\nnll\label{line1:response-true}\>\textbf{return} {\true}\\ 
\\
Upon \textit{read}$(a)$\\
\nnll\label{line1:read:snapshot}\> $S := AS.\textit{snapshot}()$\\ 
\nnll\label{line1:read:response}\> \textbf{return} $\textit{balance}(a,S)$\\

\end{tabbing}
 }
 \hrule
\caption{\namesh: Wait-free implementation of \textsf{asset-transfer} with one owner per account: code for process $p$}
\label{fig:waitfree}
\end{figure}

\begin{theorem}
  \label{th:waitfree}
  The \textsf{asset-transfer} object with a single owner per account type has a wait-free implementation in the
  read-write shared memory model.
\end{theorem}  
\begin{proof}
  Fix an execution $E$ of the algorithm in Figure~\ref{fig:waitfree}.
   Atomic snapshots can be wait-free implemented in the read-write
   shared memory model~\cite{AADGMS93}. 
  As every operation only involves a finite number of atomic snapshot
  accesses, every process completes each of the operations it invokes
  in a finite number of its own steps.
  
Let $\textit{Ops}$ be the set of:
\begin{itemize}
\item All invocations of \transfer\ or \emph{read} in $E$ that returned, and
\item All invocations of \transfer\ in $E$ that completed the \emph{update} operation (line~\ref{line1:tsf:update}) (the atomic snapshot operation has been linearized).
\end{itemize}

Let $H$ be the history of $E$. 
We define a completion of $H$ and, for each $o\in\textit{Ops}$, we define a linearization point as follows:

\begin{itemize}

  \item If $o$ is a {\read} operation, it linearizes at the
    linearization point of the \emph{snapshot} operation in
    line~\ref{line1:read:snapshot}.

  \item   If $o$ is a {\transfer} operation that returns {\false},
    it linearizes at the linearization point of the \emph{snapshot} operation in
    line~\ref{line1:tsf:snapshot}.
    
 \item   If $o$ is a {\transfer} operation that completed the \emph{update} operation,  it linearizes at the
    linearization point of the \emph{update} operation in
    line~\ref{line1:tsf:update}.
    If $o$ is incomplete in $H$, we complete it with response $\true$.

\end{itemize}  

Let $\bar H$ be the resulting complete history and let $L$ be the sequence
of complete invocations of $\bar H$ in the order of their
linearization points in $E$.
Note that, by the way we linearize invocations, the linearization of a
prefix of $E$ is a prefix of $L$. 

Now we show that $L$ is legal and, thus, $H$ is
linearizable.
We proceed by induction, starting with the empty (trivially legal)
prefix of $L$.
Let $L_{\ell}$ be the legal prefix of the first $\ell$ invocations and 
$op$ be the $(\ell+1)$st operation of $L$.
Let $op$ be invoked by process $p$.
The following cases are possible:

\begin{itemize}

\item $op$ is a {\read}$(a)$: the snapshot taken at the linearization point of $op$
  contains all successful transfers concerning $a$ in $L_{\ell}$. By
  the induction hypothesis, the resulting balance is non-negative.  

\item $op$ is a failed {\transfer}$(a,b,x)$: the snapshot taken at the linearization point of $op$
  contains all successful transfers concerning $a$ in $L_{\ell}$. By
  the induction hypothesis, the resulting balance is non-negative. 
  
\item  $op$ is a successful  {\transfer}$(a,b,x)$:  by the algorithm,
  before the linearization point of $op$, process $p$ took a snapshot.
  Let $L_{k}$, $k\leq\ell$, be the prefix of $L_{\ell}$ that only
  contain operations linearized before the point in time when the 
  snapshot was taken by $p$.  
  
  We observe that $L_k$ includes a \emph{subset} of all incoming transfers on $a$ and
  \emph{all} outgoing transfers on $a$ in $L_{\ell}$. Indeed, as $p$
  is the owner of $a$ and only the owner of $a$ can perform outgoing
  transfers on $a$, all outgoing transfers in $L_{\ell}$ were
  linearized before the moment $p$ took the snapshot within $op$.
  Thus, $\textit{balance}(a,L_k) \leq \textit{balance}(a,L_{\ell})$.%
  \footnote{
    Analogously to $\textit{balance}(a,S)$ that computes the balance for account $a$ based on the transfers contained in snapshot $S$,
    $\textit{balance}(a,L)$, if $L$ is a sequence of operations, computes the balance of account $a$ based on all transfers in $L$.
  }
  
  By the algorithm, as $op={\transfer}(a,b,x)$ succeeds, we have  
  $\textit{balance}(a,L_k)\geq x$.
  Thus, $\textit{balance}(a,L_{\ell})\geq x$ and the resulting balance
  in $L_{\ell+1}$ is non-negative.

\end{itemize}    

Thus, $H$ is linearizable. 
\end{proof}  

\begin{corollary}\label{cor:consensus}
The \textsf{asset-transfer} object type with one owner per account has consensus number $1$.
\end{corollary}


\section{\name in the Message Passing Model}
\label{sec:algo}

In this section we abandon the crash-stop shared memory model that we only used for reasoning about the consensus number of the \textsf{asset-transfer} data type.
Instead, we consider a distributed system where processes communicate by sending messages through authenticated channels
and where a limited fraction of processes may behave in an arbitrary (Byzantine) manner.
In this setting, Byzantine processes may attempt to double-spend, i.e., initiate
transfers which cannot be justified by the balance in their accounts.

We first refine the specification of our asset transfer abstraction to adapt it to the Byzantine message passing model.
We then present present \namemp, an algorithm that implements this abstraction.

\subsection{Asset Transfer in Message Passing with Byzantine Faults}
\label{sec:def-byz}

We now refine our asset transfer specification for the Byzantine environment. We only require
that the transfer system behaves correctly towards \emph{benign} processes,
regardless of the behavior of Byzantine ones.
We say that a process is benign if it respects the algorithm and can only fail by crashing.
\cmt{YAP}{recovery?}

Informally, we require that every system execution appears to benign processes as a
correct sequential execution.
As a result, no benign process can be a victim of
a double-spending attack.

Moreover, we require that this execution respects the real-time order among
operations performed by benign processes~\cite{Her91}: when such an
operation completes it should be visible to every future
operation invoked by a benign process.

For the sake of efficiency, in our algorithm (\cref{sec:algo}), we
slightly relax the last requirement---while still preventing
double-spending.
We require that \emph{successful} \operation{\transfer}{} operations invoked
by benign processes constitute a legal sequential history that preserves the real-time
order.
A \operation{read}{} or a failed \operation{transfer}{} operation invoked
by a benign process $p$ can be
``outdated''---it can be based on a stale state of $p$'s balance.
Informally, one can view the system requirements as \emph{linearizability}~\cite{herl90linearizability}
for successful transfers and \emph{sequential consistency}~\cite{Attiya1994} for failed
transfers and reads.
As progress (liveness) guarantees, we require that every operation
invoked by a correct process eventually completes.
\cmt{YAP}{eventual completition might not be considered good enough in some cases.. can we say more about the completition time, at least under additional assumptions?}
\cmt{<--MP}{In a completely asynchronous system, we cannot argue about any real-world time, but probably we can say something about the number of exchanged messages.}
One can argue that this relaxation incurs little impact on the system's utility, as long as all incoming transfers
are eventually applied.

\subsection{\namemp\ Overview}

We now present \namemp, a decentralized algorithm that implements secure transfers in a message passing system subject to Byzantine faults.
Instead of consensus, \namemp relies on a causally ordered broadcast which precisely captures the semantics of the transfer applications.

\Cref{fig:algo-modular} depicts the high-level modules of \namemp.
There are two main modules: one for tracking dependencies (i.e., the applied incoming transfers), plus an underlying secure broadcast protocol.
For simplicity, in this section we present only the basic transfer algorithm and assume an existing secure broadcast protocol as a black box (which we will describe later).
Depending on the exact system model and the properties of this broadcast protocol, we obtain two variants of \namemp.
Concretely, we consider a \emph{deterministic} secure broadcast algorithm, which we use to obtain \emph{\namedet} (\Cref{sec:algo-det}), and a \emph{probabilistic} secure broadcast protocol, which underlies \emph{\nameprob} (\Cref{sec:algo-prob}).

\begin{figure}[ht!]
\includegraphics[width=0.35\columnwidth]{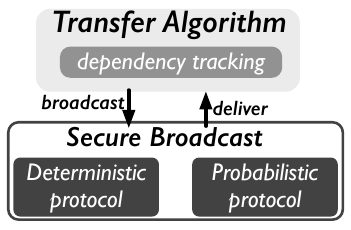}
\caption{High-level design of our asset transfer system. Depending on the underlying secure broadcast primitive, there are two variants: \emph{\namedet} works in the deterministic model, while \emph{\nameprob} has probabilistic guarantees.}
\label{fig:algo-modular}
\end{figure}


Our algorithm works as follows.
To perform a transfer, a process $p$ broadcasts a message with the
transfer details.
These details match the arguments of the \operation{transfer}{} operation (see~\Cref{sec:def-sequential}), and consist of the outgoing account (in this case, the account of $p$), the incoming account, and the transferred amount.

To ensure the authenticity of operations---so that no process
is able to debit another process's account---we assume that
processes sign all their messages before broadcasting them.
In practice, similar to Bitcoin and other transfer systems,
every process possesses a public-private key pair that allows only $p$ to securely initiate transfers from its corresponding account.

Recall that in a secure transfer system, a correct process should not
accept a transfer $P$ before it accepts all transfers $P$ depends
  on.
In particular, transfers originating at the same process $p$ do not commute,
and thus must be delivered by all correct processes in the same order.
This property is known as the \emph{source order}, and we ensure it by relying on secure broadcast~\cite{ma97secure}.

Secure broadcast, by itself, is not sufficient to ensure that causality is preserved, because there is an additional scenario in which two transfers do not commute.
Suppose that process $p$ has initially zero balance and receives
money in a transfer $P_1$.
Thereafter, process $p$ immediately sends money to process $q$ in a transfer $P_2$.
If process $q$ tries to apply $P_2$ before being aware of $P_1$, then $q$ would consider the former transfer invalid due to insufficient funds.
Thus, $P_2$ depends on $P_1$ and all processes in the system should apply $P_1$ before $P_2$.
We enforce this ordering by attaching a dependency set---called a
\emph{history}---to every transfer.
Before delivering a transfer, we require every process to deliver the history attached to that transfer.
Note that in practice, this history need not contain the full details of other transfers, but merely their identifiers (in a similar vein as vector clocks \cite{fidgevc}).
\cmt{YAP}{explain in more detail, either here or later. Mention where the full details of transfers are stored, how they can be accessed despite churn, the amount of history that is necessary, etc...}
\cmt{MP}{Agree that this is important for practice, and it might become quite tricky to do this efficiently.}

\begin{figure}
\hrule \vspace{1mm}
 {\small
\begin{tabbing}
 bbb\=bb\=bb\=bb\=bb\=bb\=bb\=bb \=  \kill
\textbf{Local variables}: \\
$\textit{seq}[\textit{ }]$, initially $\textit{seq}[q] = 0$, $\forall q$ \emph{\{Number of validated transfers outgoing from $q$\}}\\
$\textit{rec}[\textit{ }]$, initially $\textit{rec}[q] = 0$, $\forall
q$ \emph{\{Number of delivered transfers from $q$\}}
\\
$hist[\textit{ }]$, initially $hist[q] = \emptyset$, $\forall q$ \emph{\{Set of validated transfers involving $q$\}}
\\
$\textit{deps}$, initially $\emptyset$ \emph{\{Set of last incoming transfers for account of local process p\}}
\\
$\textit{toValidate}$, initially $\emptyset$ \emph{\{Set of delivered (but not validated) transfers\}}
\\
\\
\nnll{}{} \textbf{operation} \textit{transfer}$(a,b,x)$ where $a=p$
\commentline{Transfer an amount of $x$ from account $a$ to account $b$}\\
\nnll\label{line:check-balance}\> \textbf{if} $\textit{balance}(a,hist[p]\cup\textit{deps})< x$
\textbf{then}\\
\nnll\label{line:response-false}\>\> \textbf{return} \textit{false}\\
\nnll\label{line:check-sbcast}\>
$\textit{broadcast}([(a,b,x,\textit{seq}[p]+1),\textit{deps}])$\\ 
\nnll\label{line:dep-null}\> $\textit{deps} :=\emptyset$\\
\\
\nnll{}{} \textbf{operation} \textit{read}$(a)$
\commentline{Read balance of account a}\\
\nnll\label{line:response-read}\> \textbf{return} $\textit{balance}(a,hist[a]\cup\textit{deps})$\\
\\
\commentsimple{Secure broadcast callback}\\
\nnll{}{}\label{line:deliver} \textbf{upon} $\textit{deliver}(q,m)$
\commentline{Executed when  $p$ delivers message $m$ from process $q$}\\
\nnll\label{line:sb-deliver}\> let $m$ be $[(q,d,y,s),h]$\\
\nnll\label{line:nextrec}\> \textbf{if} $s=\textit{rec}[q]+1$ \textbf{then} \\
\nnll\>\> $\textit{rec}[q]:=\textit{rec}[q]+1$\\
\nnll\>\> $\textit{toValidate} :=  \textit{toValidate}\cup\{(q,m)\}$\\

\\
\nnll\label{line:validate}{}{} \textbf{upon} $(q,[t,h])\in \textit{toValidate} \;\wedge \;\textit{Valid}(q,t,h)$
\commentline{Executed when a transfer delivered from $q$ becomes valid}\\
\nnll\> let $t$ be $(c,d,y,s)$\\
\nnll\label{line:append-outgoing}\>\> $hist[c] := hist[c]\cup t$
\commentline{Update the history for the outgoing account c}\\
\nnll\label{line:append-incoming}\>\> $hist[d] := hist[d]\cup t$
\commentline{Update the history for the incoming account d}\\
\nnll\>\> $\textit{seq}[q] := s$\\
\nnll\>\> \textbf{if} $d=p$ \textbf{then} \\
\nnll\label{line:deps}\>\>\> $\textit{deps} := \textit{deps}\cup (c,d,y,s)$
\commentline{This transfer is incoming to account of local process p}
\\
\nnll\>\> \textbf{if} $c=p$ \textbf{then} \\
\nnll\label{line:response-true}\>\>\> \textbf{return} \textit{true}
\commentline{This transfer is outgoing from account of local process p}\\
\\
\nnll\label{line:valid}{}{} \textbf{function} $\textit{Valid}(q,t,h)$\\
\nnll\> let $t$ be $(c,d,y,s)$\\
\nnll\label{line:validation}\> \textbf{return} ($q=c$)\\
\nnll\label{line:check-hist1}\>\>\>\>\textbf{and} ($s= \textit{seq}[q]+1$)\\
\nnll\label{line:check-hist2}\>\>\>\>\textbf{and} ($\textit{balance}(c,hist[q])\geq y$)\\
\nnll\label{line:check-hist3}\>\>\>\>\textbf{and} ($h\subseteq hist[q]$)\\

\\
\nnll{}{} \textbf{function} $\textit{balance}(a,h)$\\
\nnll\label{line:balance}\> \textbf{return} sum of incoming transfers minus outgoing transfers for account $a$ in $h$
\end{tabbing}
 }
 \hrule
\caption{\namemp{}: an algorithm for a consensusless transfer system based on secure broadcast. Code for every process $p$.}

\label{fig:banking-relaxed}
\end{figure}


\subsection{The \namemp\ Algorithm}
\label{sec:transfer-algo}

\Cref{fig:banking-relaxed} describes our \namemp\ consensusless algorithm
implementing a transfer system (as defined in \cref{sec:def-byz}).
Each process $p$ maintains, for each process $q$, an integer
$\textit{seq}[q]$ reflecting the number of transfers which
process $q$ initiated and which process $p$ has validated and applied.
Process $p$ also maintains, for every process $q$, an integer
$\textit{rec}[q]$ which reflects the number of transfers which process $q$ has initiated and process $p$ has delivered (but not necessarily applied).

Additionally, there is also a list $\textit{hist}[q]$ of transfers which \emph{involve} process $q$.
We say that a transfer operation involves a process $q$ if that transfer is either outgoing or incoming on the account of $q$.
Each process $p$ maintains as well a local variable $\textit{deps}$.
This is a set of transfers \emph{incoming} for $p$ that $p$ has applied since  the last successful (\emph{outgoing}) transfer.
Finally, the set $\textit{toValidate}$ contains delivered transfers that are pending validation (i.e., have not been validated nor applied).

To perform a $\textit{transfer}$ operation, process $p$ first checks the
balance of its own account, and if there is not enough
funding, i.e., the balance is insufficient, returns $\textit{false}$ (line~\ref{line:response-false}).
Otherwise, process $p$ broadcasts a message with this operation via the secure broadcast primitive (line~\ref{line:check-sbcast}).
This message includes the three basic arguments of a \operation{transfer}{} operation as well as $\textit{seq}[p]+1$ and dependencies $\textit{deps}$.
Each correct process in the system eventually delivers this message via a callback from secure broadcast (line~\ref{line:deliver}).
Upon delivery, process $p$ checks this message for well-formedness (lines~\ref{line:sb-deliver} and~\ref{line:nextrec}), and then adds it to the set of messages pending validation.
We explain the validation procedure later.

Once a transfer passes validation (the predicate in line~\ref{line:validate} is satisfied), process $p$ applies this transfer on the local state.
Applying a transfer means that process $p$ adds this transfer to both the history of the outgoing (line~\ref{line:append-outgoing}) and incoming accounts (line~\ref{line:append-incoming}).
If the transfer is incoming for local process $p$, it is also added to $\textit{deps}$, the set of current dependencies for $p$ (line~\ref{line:deps}).
If the transfer is outgoing for $p$, i.e., it is the currently pending
\operation{transfer}{} operation invoked by $p$, then the response $\textit{true}$ is returned (line~\ref{line:response-true}).

To perform a \textit{read}$(a)$ operation for account $a$, process $p$ simply computes the
balance of this account based on the local history $hist[a]$ (line~\ref{line:balance}).

Before applying a transfer $op$ from some process $q$, process $p$ validates $op$ via the \emph{Valid} function (lines~\ref{line:valid}--\ref{line:check-hist3}).
To be valid, $op$ must satisfy four conditions.
The first condition is that process $q$ (the issuer of transfer $op$) must be the owner of the outgoing account for $op$ (line~\ref{line:validation}).
Second, any preceding transfers that process $q$ issued must have been validated (line~\ref{line:check-hist1}).
Third, the balance of account $q$ must not drop below zero (line~\ref{line:check-hist2}).
Finally, the reported dependencies of $op$ (encoded in $h$ of
line~\ref{line:check-hist3}) must have been validated and exist in
$\textit{hist}[q]$.



\section{\namedet: the Deterministic Case}
\label{sec:algo-det}

\namedet implements a transfer algorithm for the deterministic system model.
It builds on a secure broadcast algorithm which assumes an asynchronous network of $N$ processes, where less than $N/3$ of processes can be Byzantine.
We briefly discuss this broadcast algorithm here\ignore{with further details deferred to Appendix~\ref{app:secure-broadcast})}, while the rest of the transfer algorithm is identical to the one we described earlier (\Cref{{fig:banking-relaxed}}).

\subsection{Deterministic Secure Broadcast}
\label{sec:sb}

Secure broadcast for the deterministic model has its roots in the Asynchronous Byzantine Agreement (ABA) problem, defined by Bracha and Toueg~\cite{br85acb}.
ABA is a single-shot abstraction for agreeing on the content of a single message broadcast from a designated sender.
Informally, secure broadcast is a multi-shot version of ABA, and guarantees that messages broadcast by a given (correct) sender are delivered by all correct processes in the same order.


There are multiple algorithms for implementing secure broadcast in deterministic system models, e.g., the double-echo algorithm, initially described by Bracha~\cite{bra87asynchronous} as a single-shot version (which appears in several other works, including practical systems~\cite{cac11intro,cach02sintra,du18beat}), as well as the \emph{secure reliable multicast} of Reiter~\cite{re94ramp}, which relies on digital signatures and which was optimized in several aspects~\cite{ma97secure,MR97srm}.

Below we sketch a protocol for secure broadcast which we use in \namedet.
This protocol is not novel (unlike the probabilistic version we use for \nameprob in~\Cref{sec:algo-prob}), and it draws directly from the signature-based algorithm due to Malkhi and Reiter~\cite{MR97srm}.
We chose to use this protocol for its simplicity.
This description uses an underlying reliable broadcast primitive~\cite{cac11intro}, and we use the terms \emph{reliable-broadcast} and \emph{reliable-deliver} to denote the invocation and callback of this primitive.

To broadcast a message $m$ securely, a process $s$ attaches a sequence number $i$ to that message and then disseminates the tuple $(m, i)$ using reliable-broadcast.
Upon reliable-delivery of $(m, i)$, every process $p$ checks whether this is the first time it sees sequence number $i$ from process $s$. If yes, process $p$ replies directly to $s$ with a signed acknowledgment of the tuple $(m, i)$.
When process $s$ receives acknowledgments from a quorum of more than two thirds of the processes, $s$ broadcasts the set of obtained acknowledgments using reliable-broadcast.
Any process can deliver $m$ after obtaining the correct set of acknowledgments for $(m, i)$, and after having delivered all messages from $s$ that have sequence numbers smaller than $i$.

Intuitively, this algorithm ensures that any two delivered
  messages are ``witnessed'' by at least one correct process.
  This way we ensure that messages from the same source, whether benign
  or Byzantine, are delivered by benign processes in the same order.




\section{\nameprob: the Probabilistic Case}
\label{sec:algo-prob}

\nameprob implements an asset transfer algorithm in a probabilistic system model, by building on a novel secure broadcast algorithm with probabilistic guarantees.
In the rest of this section we first discuss our system model more precisely (\Cref{section:prob-model}).
We then present the secure broadcast algorithm in three steps.
First, we introduce an abstraction for broadcasting a single message with probabilistic delivery guarantees (\Cref{section:broadcast}).
Second, we refine this abstraction to make it \emph{consistent}, ensuring that processes deliver the \emph{same} single message (\Cref{section:singleshot}).
Third, we obtain probabilistic secure broadcast as a multi-shot abstraction over consistent broadcast (\Cref{section:multishot}).


\subsection{Preliminaries}
\label{section:prob-model}

We discuss the system model, including assumptions we make on the network and on the information available both to correct and Byzantine processes.
We also introduce notation that will be valid throughout the rest of this section.

For all our probabilistic algorithms, we have the following assumptions:
\begin{enumerate}
    \item \label{ergpdeassumption:processes} (\textbf{Processes}) The set $\Pi$ of processes partaking in the algorithm is fixed. Unless stated otherwise, we let $N = |\Pi|$ denote the number of processes, and refer to the $i$-th process as $\pi_i \in \Pi$.
    \item \label{ergpdeassumption:links} (\textbf{Links}) Any two processes can communicate via reliable, authenticated, point-to-point links~\cite{cac11intro}.
    \item \label{ergpdeassumption:failures} (\textbf{Failures}) At most a fraction $f$ of the processes is Byzantine, i.e., subject to arbitrary failures. Byzantine processes are under the control of the same adversary, and can take coordinated action. We also assume that the adversary does not have access to the output of local randomness sources of correct processes.
    Unless stated otherwise, we let $\Pi_C \subseteq \Pi$ denote the set of correct processes and $C = \abs{\Pi_C} = \lceil\rp{1 - f} N\rceil$ denote the number of correct processes.
    \item \label{ergpdeassumption:asynchrony} (\textbf{Asynchrony}) Byzantine processes can cause arbitrary but finite delays on any link, including links between pairs of correct processes.
    \item \label{ergpdeassumption:anonymity} (\textbf{Anonimity}) Byzantine processes cannot determine which correct processes another correct process is communicating with.
    \item \label{ergpdeassumption:sampling} (\textbf{Sampling}) Every process has direct access to an oracle $\Omega$ that, provided with an integer $n \leq N$, yields the public keys of $n$ distinct processes, chosen uniformly at random from $\Pi$.
\end{enumerate}

We later weaken assumption (\ref{ergpdeassumption:processes}) into an inequality, as we generalize our results to systems with slow churn.
Assumption (\ref{ergpdeassumption:asynchrony}) represents one of the main strengths of this work: messages can be delayed arbitrarily and maliciously without compromising the security property of any of the algorithms presented in this work.
Assumption (\ref{ergpdeassumption:anonymity}) represents the strongest constraint we put on the knowledge of the adversary. We later show that, without this assumption, an adversary could easily poison the view of the system of a targeted correct process without having to interfere with any local randomness source.
Even against ISP-grade adversaries, assumption (\ref{ergpdeassumption:anonymity}) can be implemented in practice by means of, e.g., onion routing~\cite{syv04tor} or private messaging~\cite{van15vuvu} algorithms.
Assumption (\ref{ergpdeassumption:sampling}) reduces, in the permissioned case, to randomly sampling an exhaustive list of processes.
We later discuss how a membership sampling algorithm can be used, in conjunction with Sybil resistance strategies, to implement the sampling oracle in the permissionless case (see \cref{section:permissionless}).


\subsection{\Broadcast}
\label{section:broadcast}

In this section, we introduce the \emph{\broadcast{}} abstraction and discuss its properties.
This abstraction serves the purpose of reliably broadcasting a single message from a designated (correct) sender to all correct processes.
We then present \erg, a probabilistic algorithm that implements \broadcast, and evaluate its security and complexity as a function of its parameters.
We use \broadcast\ in the implementation of \pde\ (see \Cref{section:singleshot}) to initially distribute a message from the designated sender to all correct processes.

\subsubsection{Definition}

The \broadcast{} interface (instance $\broadcastinstance$, sender $\sigma$) exposes the following two events:
\begin{itemize}
    \item \textbf{Request}: $\event{\broadcastinstance}{Broadcast}{m}$: Broadcasts a message $m$ to all processes. This is only used by $\sigma$.
    \item \textbf{Indication} $\event{\broadcastinstance}{Deliver}{m}$: Delivers a message $m$ broadcast by process $\sigma$.
\end{itemize}

For any $\epsilon \in [0, 1]$, \Broadcast\ is $\epsilon$-secure if:
\begin{enumerate}
    \item \textbf{No duplication}: No correct process delivers more than one message.
    \item \textbf{Integrity}: If a correct process delivers a message $m$, and $\sigma$ is correct, then $m$ was previously broadcast by $\sigma$.
    \item $\epsilon$-\textbf{Validity}: If $\sigma$ is correct, and $\sigma$ broadcasts a message $m$, then $\sigma$ eventually delivers $m$ with probability at least $(1 - \epsilon)$.
    \item $\epsilon$-\textbf{Totality}: If a correct process delivers a message, then every correct process eventually delivers a message with probabiity at least $(1 - \epsilon)$.
\end{enumerate}

\subsubsection{Algorithm}
\label{subsection:ergalgorithm}

\begin{algorithm}
\begin{algorithmic}[1]
\Implements
    \Instance{\broadcastabstraction}{\broadcastinstance}
\EndImplements

\Uses
    \Instance{AuthenticatedPointToPointLinks}{al}
\EndUses

\Parameters
    \State $G$: expected gossip sample size
\EndParameters

\sUpon{pb}{Init}
    \State $\mathcal{G} = \Omega(poisson(G))$; \label{line:erginitializesample}
    \ForAll{\pi}{\mathcal{G}}
        \Trigger{al}{Send}{\pi, [\text{\tt GossipSubscribe}]}; \label{line:ergsubscribe}
    \EndForAll
    \State $delivered = \bot$;
\EndsUpon

\Upon{al}{Deliver}{\pi, [\text{\tt GossipSubscribe}]} \label{line:ergreceivesubscribe}
    \If{delivered \neq \bot}
        \State $(message, signature) = delivered$;
        \Trigger{al}{Send}{\pi, [\text{\tt Gossip}, message, signature]}; \label{line:ergcatchup}
    \EndIf
    \State $\mathcal{G} \leftarrow \mathcal{G} \cup \{\pi\}$; \label{line:ergupdatesample}
\EndUpon

\Procedure{dispatch}{message, signature}
    \If{delivered = \bot} \label{line:ergcheckdelivered}
        \State $delivered \leftarrow (message, signature)$; \label{line:ergsetdelivered}
        \ForAll{\pi}{\mathcal{G}}
            \Trigger{al}{Send}{\pi, [\text{\tt Gossip}, message, signature]}; \label{line:ergforward}
        \EndForAll
        \Trigger{\broadcastinstance}{Deliver}{message} \label{line:ergdeliver}
    \EndIf
\EndProcedure

\Upon{pb}{Broadcast}{message} \Comment{only process $\sigma$}
    \State $dispatch(message, sign(message))$; \label{line:ergbroadcast}
\EndUpon

\Upon{al}{Deliver}{\pi, [\text{\tt Gossip}, message, signature]}
    \If{verify(\sigma, message, signature)} \label{line:ergchecksignature}
        \State $dispatch(message, signature)$;
    \EndIf
\EndUpon

\end{algorithmic}
\caption{\erg}
\label{algorithm:erg}
\end{algorithm}

\Cref{algorithm:erg} lists the implementation of \erg.
This algorithm distributes a single message\footnote{Note that one instance of \broadcast\ only distributes a single message. To disseminate multiple messages, we use multiple instances of \broadcast.} across the system by means of gossip: upon reception, a correct process relays the message to a set of randomly selected neighbors.
The algorithm depends on one integer parameter, $G$ (expected gossip sample size), whose value we study in \cref{appendix:erganalysis}.

\paragraph{Initialization} Upon initialization, (line \ref{line:erginitializesample}) a correct process randomly samples a value $\bar G$ from a \emph{Poisson} distribution with expected value $G$, and uses the sampling oracle $\Omega$ to select $\bar G$ distinct processes that it will use to initialize its gossip sample $\mathcal{G}$.

\paragraph{Link reciprocation} Once its gossip sample is initialized, a correct process sends a {\tt GossipSubscribe} message to all the processes in $\mathcal{G}$ (line \ref{line:ergsubscribe}). Upon receiving a {\tt GossipSubscribe} message from a process $\pi$ (line \ref{line:ergreceivesubscribe}), a correct process adds $\pi$ to its own gossip sample $\mathcal{G}$ (line \ref{line:ergupdatesample}), and sends back the gossiped message if it has already received it (line \ref{line:ergcatchup}).

\paragraph{Gossip} When broadcasting the message (line \ref{line:ergbroadcast}), a correct designated sender $\sigma$ signs the message and sends it to every process in its sample $\mathcal{G}$ (line \ref{line:ergforward}). Upon receiving a correctly signed message from $\sigma$ (line \ref{line:ergchecksignature}) for the first time (this is enforced by updating the value of $delivered$, line \ref{line:ergcheckdelivered}), a correct process delivers it (line \ref{line:ergdeliver}) and forwards it to every process in its gossip sample (line \ref{line:ergforward}).

For a discussion on the correctness of \erg, we refer the interested reader to~\Cref{appendix:erganalysis}.


\subsection{\Singleshot}
\label{section:singleshot}

In this section, we introduce the \emph{\singleshot{}} abstraction and discuss its properties. We then present \pde, a probabilistic algorithm that implements \singleshot, and evaluate its security and complexity as a function of its parameters.

The \singleshot\ abstraction allows the set of correct processes to agree on a message from a designated sender, potentially Byzantine. \Singleshot\ is a strictly stronger abstraction than \broadcast. \Broadcast\ only guarantees that, if the sender is correct, all correct processes deliver its message and, if any correct process delivers a message, every correct process delivers a message. \Singleshot\ also guarantees that, even if the sender is Byzantine, no two correct processes will deliver different messages.

We use \singleshot\ in the implementation of \tfayto\ (see \cref{section:multishot}) as a way to consistently broadcast sequenced messages.

\subsubsection{Definition}

The \singleshot{} interface (instance $\singleshotinstance$, sender $\sigma$) exposes the following two events:
\begin{itemize}
\item \textbf{Request}: $\event{\singleshotinstance}{Broadcast}{m}$: Broadcasts a message $m$ to all processes. This is only used by $\sigma$.
\item \textbf{Indication}: $\event{\singleshotinstance}{Deliver}{m}$: Delivers a message $m$ broadcast by process $\sigma$.
\end{itemize}

For any $\epsilon \in [0, 1]$, \Singleshot\ is $\epsilon$-secure if:
\begin{enumerate}
\item \textbf{No duplication}: No correct process delivers more than one message.
\item \textbf{Integrity}: If a correct process delivers a message $m$, and $\sigma$ is correct, then $m$ was previously broadcast by $\sigma$.
\item $\epsilon$-\textbf{Validity}: If $\sigma$ is correct, and $\sigma$ broadcasts a message $m$, then $\sigma$ eventually delivers $m$ with probability at least $(1 - \epsilon)$.
\item $\epsilon$-\textbf{Totality}: If a correct process delivers a message, then every correct process eventually delivers a message with probabiity at least $(1 - \epsilon)$.
\item $\epsilon$-\textbf{Consistency}: Every correct process that delivers a message delivers the same message with probability at least $(1 - \epsilon)$.
\end{enumerate}

\subsubsection{Algorithm}
\label{subsection:pdealgorithm}

\begin{algorithm}
\begin{algorithmic}[1]
\Implements
    \Instance{\singleshotabstraction}{\singleshotinstance}
\EndImplements

\Uses
    \Instance{AuthenticatedPointToPointLinks}{al}
    \Instance{\broadcastabstraction}{\broadcastinstance}
\EndUses

\Parameters
    \State $E$: echo sample size \tabto*{5cm} $\hat E$: ready threshold
    \State $R$: ready sample size \tabto*{5cm} $\hat R$: feedback threshold
    \State $D$: delivery sample size \tabto*{5cm} $\hat D$: delivery threshold
\EndParameters

\Procedure{sample}{message, size}
    \State $\psi = \emptyset$;
    \Ntimes{size}
        \State $\psi \leftarrow \psi \cup \Omega(1)$; \label{line:pdesampleselection}
    \EndNtimes
    \ForAll{\pi}{\psi}
        \Trigger{al}{Send}{\pi, [message]}; \label{line:pdesamplesubscribe}
    \EndForAll
    \State \Return $\psi$;
\EndProcedure

\sUpon{\singleshotinstance}{Init} \label{line:pdeinitialization}
    \State $echo = \bot$; \tabto*{3.5cm} $ready = \bot$; \tabto*{7cm} $delivered = \false$;
    \State
    \State $\mathcal{E} = sample(\text{\tt EchoSubscribe}, E)$; \tabto*{6.5cm} $replies.echo = \{\bot\}^E$;
    \State $\mathcal{R} = sample(\text{\tt ReadySubscribe}, R)$; \tabto*{6.5cm} $replies.ready = \{\bot\}^R$;
    \State $\mathcal{D} = sample(\text{\tt ReadySubscribe, D})$; \tabto*{6.5cm} $replies.delivery = \{\bot\}^D$;
    \State
    \State $\tilde{\mathcal{E}} = \emptyset$; \tabto*{2cm} $\tilde{\mathcal{R}} = \emptyset$;
\EndsUpon

\Upon{al}{Deliver}{\pi, [\text{\tt EchoSubscribe}]}
    \If{echo \neq \bot}
        \State $(message, signature) = echo$;
        \Trigger{al}{Send}{\pi, [{\tt Echo}, message, signature]}; \label{line:pdeechocatchup}
    \EndIf
    \State $\tilde{\mathcal{E}} \leftarrow \tilde{\mathcal{E}} \cup \{\pi\}$; \label{line:pdeechosubscribe}
\EndUpon
\algstore{pde}
\end{algorithmic}
\caption{\pde}
\label{algorithm:pde}

\end{algorithm}

\begin{algorithm}
\begin{algorithmic}[1]
\algrestore{pde}

\Upon{al}{Deliver}{\pi, [\text{\tt ReadySubscribe}]}
    \If{ready \neq \bot}
        \State $(message, signature) = ready$;
        \Trigger{al}{Send}{\pi, [{\tt Ready}, message, signature]}; \label{line:pdereadycatchup}
    \EndIf
    \State $\tilde{\mathcal{R}} \leftarrow \tilde{\mathcal{R}} \cup \{\pi\}$; \label{line:pdereadysubscribe}
\EndUpon

\Upon{\singleshotinstance}{Broadcast}{message} \Comment{only process $\sigma$}
    \Trigger{\broadcastinstance}{Broadcast}{[\text{\tt Send}, message, sign(message)]}; \label{line:pdebroadcast}
\EndUpon

\Upon{\broadcastinstance}{Deliver}{[\text{\tt Send}, message, signature]} \label{line:pdepbdelivery}
    \If{verify(\sigma, message, signature)}
        \State $echo \leftarrow (message, signature)$;
        \ForAll{\rho}{\tilde{\mathcal{E}}}
            \Trigger{al}{Send}{\rho, [\text{\tt Echo}, message, signature]}; \label{line:pdesendecho}
        \EndForAll
    \EndIf
\EndUpon

\Upon{al}{Deliver}{\pi, [\text{\tt Echo}, message, signature]}
    \If{\pi \in \mathcal{E}\; \textbf{and}\; verify(\sigma, message, signature)\; \textbf{and}\; replies.echo[\pi] = \bot} \label{line:pdeechomessagecheck}
        \State $replies.echo[\pi] \leftarrow (message, signature)$;
    \EndIf
\EndUpon

\UponExists{message}{|\{\rho \in \mathcal{E} \mid replies.echo[\rho] = (message, signature)\}| \geq \hat E\; \textbf{and}\; ready = \bot} \label{line:pdereadyfromecho}
    \State $ready \leftarrow (message, signature)$;
    \ForAll{\rho}{\tilde{\mathcal{R}}}
        \Trigger{al}{Send}{\rho, [\text{\tt Ready}, \sigma, message, signature]}; \label{line:pdesendreadyfromecho}
    \EndForAll
\EndUponExists

\Upon{al}{Deliver}{\pi, [\text{\tt Ready}, message, signature]}
    \If{verify(\sigma, message, signature)}
        \If{\pi \in \mathcal{R}\; \textbf{and}\; replies.ready[\pi] = \bot} \label{line:pdereadymessagesourcecheck}
            \State $replies.ready[\pi] \leftarrow (message, signature)$;
        \EndIf
        \If{\pi \in \mathcal{D}\; \textbf{and}\; replies.delivery[\pi] = \bot} \label{line:pdedeliverymessagesourcecheck}
            \State $replies.delivery[\pi] \leftarrow (message, signature)$;
        \EndIf
    \EndIf
\EndUpon

\algstore{pde}
\end{algorithmic}
\end{algorithm}

\begin{algorithm}
\begin{algorithmic}[1]
\algrestore{pde}

\UponExists{message}{|\{\rho \in \mathcal{R} \mid replies.ready[\rho] = (message, signature)\}| \geq \hat R\; \textbf{and}\; ready = \bot} \label{line:pdereadyfromready}
    \State $ready \leftarrow (message, signature)$;
    \ForAll{\rho}{\tilde{\mathcal{R}}}
        \Trigger{al}{Send}{\rho, [\text{\tt Ready}, \sigma, message, signature]}; \label{line:pdesendreadyfromready}
    \EndForAll
\EndUponExists

\UponExists{message}{|\{\rho \in \mathcal{D} \mid replies.delivery[\rho] = (message, signature)\}| \geq \hat D\; \textbf{and}\; delivered = \false} \label{line:pdedeliverycondition}
    \State $delivered \leftarrow \true$;
    \Trigger{\singleshotinstance}{Deliver}{message}; \label{line:pdedeliver}
\EndUponExists

\end{algorithmic}
\end{algorithm}

\Cref{algorithm:pde} lists the implementation of \pde.

This algorithm consistently distributes a single message\footnote{Note that one instance of \singleshot\ only distributes a single message. To disseminate multiple messages, we use multiple instances of \singleshot.} across the system, as follows:
\begin{itemize}
    \item Initially, \broadcast\ distributes potentially conflicting copies of the message to every correct process.
    \item Upon receiving a message $m$ from \broadcast, a correct process issues an {\tt Echo} message for $m$.
    \item When enough {\tt Echo} or {\tt Ready} messages have been collected for the same message $m$, a correct process issues a {\tt Ready} message for $m$.
    \item When enough {\tt Ready} messages have been collected for the same message $m$, a correct process delivers $m$.
\end{itemize}

A correct process collects {\tt Echo} and {\tt Ready} messages from three randomly selected samples (\emph{echo sample}, \emph{ready sample} and \emph{delivery sample}). The sizes of these samples are determined by three integer parameters ($E$, $R$ and $D$, respectively). Three additional integer parameters ($\hat E \leq E$, $\hat R \leq R$, $\hat D \leq D$) represent \emph{thresholds} to trigger the issue of {\tt Ready} messages and the delivery of the message. We discuss the values of the six parameters of \pde\ in \cref{appendix:pdeanalysis}.

\paragraph{Sampling} Upon initialization (line \ref{line:pdeinitialization}), a correct process randomly selects three samples (an \textbf{echo sample} $\mathcal{E}$ of size $E$, a \textbf{ready sample} $\mathcal{R}$ of size $R$, and a \textbf{delivery sample} $\mathcal{D}$ of size $D$). Samples are selected with replacement by repeatedly calling $\Omega$ (line \ref{line:pdesampleselection}). A correct process sends an {\tt EchoSubscribe} message to all the processes in its echo sample, and a {\tt ReadySubscribe} message to all the processes in its ready and delivery samples (line \ref{line:pdesamplesubscribe}).

\paragraph{Publish-subscribe} Unlike in the deterministic version of {\tt Authenticated Double-Echo}, where a correct process broadcasts its {\tt Echo} and {\tt Ready} messages to the whole system, here each process only listens for messages coming from its samples (lines \ref{line:pdeechomessagecheck}, \ref{line:pdereadymessagesourcecheck}, \ref{line:pdedeliverymessagesourcecheck}).

A correct process maintains an \textbf{echo subscription set} $\tilde{\mathcal{E}}$ and a \textbf{ready subscription set} $\tilde{\mathcal{R}}$.
Upon receiving a {\tt Subscribe} message from a process $\pi$, a correct process adds $\pi$ to $\tilde{\mathcal{E}}$ (in case of {\tt EchoSubscribe}, line \ref{line:pdeechosubscribe}) or to $\tilde{\mathcal{R}}$ (in case of {\tt ReadySubscribe}, line \ref{line:pdereadysubscribe}). If a correct process receives a {\tt Subscribe} message \emph{after} publishing an {\tt Echo} or a {\tt Ready} message, it also sends back the previously published message (line \ref{line:pdeechocatchup},\ref{line:pdereadycatchup}).
A correct process will only send its {\tt Echo} and {\tt Ready} messages (lines \ref{line:pdesendecho}, \ref{line:pdesendreadyfromecho}, \ref{line:pdesendreadyfromready}) to its echo and ready subscription sets respectively.

\paragraph{Echo} The designated sender $\sigma$ initially broadcasts its message using \broadcast\ (line \ref{line:pdebroadcast}). Upon \broadcastinstance.Deliver of a message $m$ (correctly signed by $\sigma$) (line \ref{line:pdepbdelivery}), a correct process sends an {\tt Echo} message for $m$ to all the nodes in its echo subscription set (line \ref{line:pdesendecho}).

\paragraph{Ready} A correct process sends a {\tt Ready} message for a message $m$ (correctly signed by $\sigma$) to all the processes in its ready subscription sample (lines \ref{line:pdesendreadyfromecho}, \ref{line:pdesendreadyfromready}) upon collecting either of:
\begin{itemize}
    \item At least $\hat E$ {\tt Echo} messages for $m$ from its echo sample (line \ref{line:pdereadyfromecho}).
    \item At least $\hat R$ {\tt Ready} messages for $m$ from its ready sample (line \ref{line:pdereadyfromready})
\end{itemize}

\paragraph{Delivery} Upon collecting at least $\hat D$ {\tt Ready} messages for the same message $m$ (correctly signed by $\sigma$) from its delivery sample (line \ref{line:pdedeliverycondition}), a correct process delivers $m$ (line \ref{line:pdedeliver}).

For an analysis of \pde, we refer the interested reader to~\Cref{appendix:pdeanalysis}.

\subsection{\Multishot}
\label{section:multishot}

In this section, we introduce the \multishot{} abstraction.
This abstraction allows correct processes to agree on a sequence of messages sent from a designated sender, potentially Byzantine.
\Multishot is a strictly stronger abstraction than \singleshot, because the former allows for an arbitrary sequence of messages to be delivered in a consistent order.
We then present \tfayto, a probabilistic algorithm that implements \Multishot .

\subsubsection{Definition}

The \multishot interface (instance $\multishotinstance$, sender $\sigma$) exposes the following two events:
\begin{itemize}
\item \textbf{Request}: $\event{\multishotinstance}{Broadcast}{m}$: Broadcasts a message $m$ to all processes. This is only used by process $\sigma$.
\item \textbf{Indication}: $\event{\multishotinstance}{Deliver}{m}$: Delivers a message $m$ broadcast by process $\sigma$.
\end{itemize}

\begin{definition}
Let $\pi$ be a correct process. Then $\pi$ \emph{initially broadcasts (or delivers)} a sequence of messages $m_1, \ldots, m_n$ if the sequence of messages it broadcasts (or delivers) begins with $m_1, \ldots, m_n$.
\end{definition}

For any $\epsilon \in [0, 1]$, we say that \multishot\ is $\epsilon$-secure if:
\begin{enumerate}
    \item \textbf{No creation}: If $\sigma$ is correct, and $\sigma$ never broadcasts more than $n$ messages, then no correct process delivers more than $n$ messages.
    \item \textbf{Integrity}: If a correct process delivers a message $m$, and $\sigma$ is correct, then $m$ was previously broadcast by $\sigma$.
    \item $\epsilon$-\textbf{Multi-validity}: If $\sigma$ is correct, and $\sigma$ initially broadcasts $m_1, \ldots, m_n$, then $\sigma$ eventually initially delivers $m_1, \ldots, m_n$ with probability at least $(1 - \epsilon)^n$.
    \item $\epsilon$-\textbf{Multi-totality}: If a correct process delivers $n$ messages, then every correct process eventually delivers $n$ messages with probability at least $(1 - \epsilon)^n$.
    \item $\epsilon$-\textbf{Multi-consistency}: Every correct process that delivers $n$ messages initially delivers the same sequence of $n$ messages with probability at least $(1 - \epsilon)^n$.
\end{enumerate}

\subsubsection{Algorithm}
\label{subsection:tfaytoalgorithm}

\begin{algorithm}
\begin{algorithmic}[1]
\Implements
    \Instance{\multishotabstraction}{\multishotinstance}
\EndImplements
\Uses
    \Instance{\singleshotabstraction}{\singleshotinstance} \Comment{multiple instances}
\EndUses

\sUpon{\multishotinstance}{Init} \label{line:tfaytoinitialization}
    \State $next = 0$;
    \State $expected = 0$;
    \State
    \State $messages[next] = \bot$; \label{line:tfaytomessagesinit}
    \State Initialize a new instance $\singleshotinstance.next$ of \singleshotabstraction; \label{line:tfaytofirstinstance}
\EndsUpon

\Procedure{expand}{{}} \label{line:tfaytoexpand}
    \State $next \leftarrow next + 1$;
    \State $messages[next] = \bot$;
    \State Initialize a new instance $\singleshotinstance.next$ of \singleshotabstraction;
\EndProcedure

\Upon{\multishotinstance}{Broadcast}{message} \Comment{only process $\sigma$} \label{line:tfaytobroadcast}
    \Trigger{\singleshotinstance.next}{Broadcast}{message}; \label{line:tfaytosingleshotbroadcast}
    \State $expand()$;
\EndUpon

\Upon{\singleshotinstance.index}{Deliver}{message} \label{line:tfaytosingleshotdelivery}
    \State $messages[index] \leftarrow message$;
\EndUpon

\cUpon{messages[expected] \neq \bot} \label{line:tfaytoexpecteddelivered}
    \Trigger{\multishotinstance}{Deliver}{messages[expected]}; \label{line:tfaytodelivery}
    \State $expected \leftarrow expected + 1$;
    \If{self \neq \sigma}
        \State $expand()$;
    \EndIf
\EndcUpon
\end{algorithmic}
\caption{\tfayto}
\label{algorithm:tfayto}
\end{algorithm}

\Cref{algorithm:tfayto} presents the implementation of \tfayto.

\tfayto\ consistently distributes across the system a sequence of messages in consistent order. It does so by using one distinct instance of \singleshotabstraction{} (\Cref{section:singleshot}) for each message in the sequence. Instances are incrementally numbered, which allows for reordering on the receiver's end.

\paragraph{Initialization} Upon initialization (line \ref{line:tfaytoinitialization}), all correct processes initialize one instance of \singleshotabstraction{} (line \ref{line:tfaytofirstinstance}), that will be used to consistently broadcast the first message. Each process also initializes an array of all the messages received through \singleshot\ (line \ref{line:tfaytomessagesinit}).

\paragraph{Broadcast} When broadcasting (line \ref{line:tfaytobroadcast}), the sender simply triggers $\singleshotinstance$.Broadcast  on the latest instance of \singleshotabstraction (line \ref{line:tfaytosingleshotbroadcast}). It then initializes new instance of \singleshotabstraction. The index of the new instance is maintained in the variable $next$, and incremented in $expand()$ (line \ref{line:tfaytoexpand}).

\paragraph{Delivery} Upon $\singleshotinstance$.Deliver (line \ref{line:tfaytosingleshotdelivery}), a correct process adds the message to the array of messages that have been received. When the next expected message (i.e., the message with the lowest sequence number which has not yet been $\multishotinstance$.Delivered) is $\singleshotinstance$.Delivered (line \ref{line:tfaytoexpecteddelivered}), it is also $\multishotinstance$.Delivered (line \ref{line:tfaytodelivery}).

\subsection{Permissionless Environment}
\label{section:permissionless}

In this section we discuss the deployment of \name\ (in particular \nameprob) in a permissionless environment.

Our implementation of \nameprob\ relies on \erg\ and \pde, both of which use an oracle $\Omega$ for obtaining random samples of processes. The size of a sample is optimized for \emph{representativeness}: this ensures, e.g., that with high probability only a bounded fraction of a sample is Byzantine.

In private / permissioned systems with full membership view, $\Omega$ can be obtained from a simple local randomness generator by picking a sample from a local list of processes.

In a public / permissionless system, however, sampling a very large, dynamically changing set of processes is non-trivial, especially in a Byzantine environment.

Moreover, in a permissionless setting, the assumption of a limited fraction of faulty processes might be broken by Sybil attacks.

To implement $\Omega$, we use an existing Byzantine-tolerant sampling protocol (Brahms \cite{bor09brahms}).
However, Brahms requires, but does not provide, a Sybil resistance mechanism.

As in all public / permissionless solutions, the problem of Sybil attacks needs to be addressed in our permissionless variant as well.
An effective solution that does not rely on any trusted admission control scheme inevitably requires linking at least some critical parts of the protocol to some real or virtual resource that is by its nature limited.
In Bitcoin, this is achieved through the (in)famous proof-of-work scheme, where computation power is used as this limited resource.
Algorand leverages the virtual currency itself to attribute ``voting power'' to participants (proof-of-stake).
Such proof-of-* mechanisms rely on a participant having a proof of her ``right to vote'' that \emph{any} other participant can easily verify.
Moreover, these mechanisms are usually an integral part of the protocol that uses them.
The way \name\ prevents Sybil attacks has two advantages over traditional methods.

First, our solution decouples the Sybil resistance mechanism from the protocol itself.
Any sub-protocol can be used to prove that participants are genuine (as opposed to Sybil identities).
Proof-of-work is just an example of a mechanism that such a sub-protocol can use.
Other mechanisms, even ones yet to be invented, are easily pluggable into our design.

Second, we relax the properties of the Sybil resistance mechanism that are necessary for it to be used with \name.
In particular, while most curent strategies are based on a \emph{globally verifiable} proof of the ownership (or at least control) of some resource,
we only require \emph{local verifiability}: only a small subset of participants needs to be able to verify a proof.
Specifically, a participant only presents proofs to those participants she is directly communicating with, and it is sufficient if these participants are able to verify these proofs.

Removing the need for global verifiability opens new doors to much more efficient Sybil resistance techniques, as we show below.

For the Sybil resistance sub-protocol, we propose a novel \emph{proof-of-bandwidth} scheme.
As the name suggests, the voting power of participants is bound to the (generally limited) resource of bandwidth.
A participant proves herself to the peers by sending them data over the network.
Participant $A$ considers another participant $B$ to be genuine (and takes into account $B$'s messages when taking decisions) only if $B$ periodically sends data to $A$.
Nodes that fail to send enough data to $A$ are disregarded by $A$.
Note that in general $B$ cannot prove to anyone else than $A$ that data has been transmitted, but in our case this is sufficient (no need for global verifiability).

The advantage of this approach is that protocol data itself can serve as such proof, and thus no additional resources need to be wasted, as is the case for proof-of-work or proof-of-storage.
In the case where the protocol itself does not generate enough data for a proof, even other completely unrelated data exchange between two participants (e.g. BitTorrent traffic) can be used as proof.
Since we only require local verifiability, it is only up to the two communicating parties to agree on what traffic constitutes a proof-of-bandwidth.
In the absolutely worst case, where no otherwise useful data can be exchanged between two participants, they can resort to exchanging garbage data.

\cmt{YAP}{- mention prob in intro\\- explain in more detail how prob and brahms can be used together}



\section{Sharing Accounts among $k$ processes}
\label{sec:k-consensus}

\cmt{RG}{We need a small intro here to say that blockchain protocols do more than asset transfers between two participants; in particular, they enable to implement contracts representing transfers of assets involving k-users. We show here that doing this does not require solving consensus among the entire set of nodes in the system.}

We now return to the crash-stop shared memory model defined in \cref{sec:result} and consider the general case with an arbitrary owner map $\mu$.
We show that an \textsf{asset-transfer} object's consensus number is the maximal number of processes sharing an account.
More precisely, the consensus number of an \textsf{asset-transfer} object is $\max_{a\in\A}|\mu(a)|$. 

We say that an \textsf{asset-transfer} object, defined on a set of
accounts $\A$ with an ownership map $\mu$,
is \emph{k-shared} iff $\max_{a\in\A}|\mu(a)|=k$.
In other words, the object is $k$-shared if $\mu$ allows at least one account to be owned by
$k$ processes, and no account is owned by more than $k$ processes.

We show that the consensus number of any \kasset\ is $k$, which
generalizes our result in \cref{sec:result}.
We first show that \kasset\ has consensus number \emph{at least $k$} by implementing consensus for $k$ processes using only registers and an instance of \kasset.
We then show that \kasset\ has consensus number \emph{at most $k$} by reducing \kasset\ to \kcons, an object known to have consensus number $k$~\cite{JT92}.


\begin{lemma}
  \label{lem:cons-to-at}
  Consensus has a wait-free implementation for $k$ processes in the read-write shared memory model equipped with a single $k$-shared \textsf{asset-transfer} object.
\end{lemma}

\begin{figure}
\hrule \vspace{1mm}
 {\small
\setcounter{linenumber}{0}
\begin{tabbing}
bbb\=bb\=bb\=bb\=bb\=bb\=bb\=bb \=  \kill

Shared variables: \\
\> $R[i], i \in 1,\ldots,k$, $k$ registers, initially $R[i] = \bot, \forall i$\\
\> $AT$, $k$-shared \textsf{asset-transfer} object containing:\\
\>\>\> -- an account $a$ with initial balance $2k$\\
\>\>\>\>owned by processes $1,\ldots,k$\\
\>\>\> -- some account $s$\\
\\
Upon \textit{propose}$(v)$:\\
\nnll\label{ln:simple-announce}\> $R[p].write(v)$\\ 
\nnll\label{ln:simple-transfer}\> $AT.transfer(a, s, 2k - p))$\\
\nnll\label{ln:simple-read}\> \textbf{return} $R[AT.read(a)].read()$

\end{tabbing}
 }
 \hrule
\caption{Wait-free implementation of consensus among $k$ processes using a \kasset. Code for process $p\in\{1,\ldots,k\}$.}
\label{fig:kcons-to-kaccount}
\end{figure}

\begin{proof}

We now provide a wait-free algorithm that solves consensus among $k$ processes using only registers and an instance of $k$-shared \textsf{asset-transfer}.
The algorithm is described in~\cref{fig:kcons-to-kaccount}.
Intuitively, $k$ processes use one shared account $a$ to elect one of them whose input value will be decided.
Before a process $p$ accesses the shared account, $p$ announces its input in a register (line~\ref{ln:simple-announce}).
Process $p$ then tries to perform a transfer from account $a$ to another account.
The amount withdrawn this way from account $a$ is chosen specifically such that:
\begin{enumerate}
\item only one transfer operation can ever succeed, and
\item if the transfer succeeds, the remaining balance on $a$ will uniquely identify process $p$.
\end{enumerate}
To satisfy the above conditions, we initialize the balance of account
$a$ to $2k$ and have each process $p\in\{1,\ldots,k\}$ transfer $2k-p$ (line
\ref{ln:simple-transfer}).
Note that  transfer operations invoked by distinct processes
$p,q\in\{1,\ldots,k\}$ have arguments $2k-p$ and $2k-q$, and
$2k-p+2k-q\geq 2k-k+2k-(k-1)=2k+1$.
The initial balance of $a$ is only $2k$ and no incoming transfers are ever executed.
Therefore, the first transfer operation to be applied to the
object succeeds (no transfer tries to withdraw more then $2k$) and the remaining operations will have to fail due to insufficient balance.
\noindent When $p$ reaches line \ref{ln:simple-read}, at least one transfer must have succeeded:
\begin{enumerate}
\item either $p$'s transfer succeeded, or
\item $p$'s transfer failed due to insufficient balance, in which case
  some other process must have previously succeeded.
\end{enumerate}
Let $q$ be the process whose transfer succeeded. Thus, the balance of account $a$ is $2k - (2k - q) = q$.
Since $q$ performed a transfer operation, by the algorithm, $q$ must
have previously written its proposal to the register $R[q]$.
Regardless of whether $p = q$ or $p \neq q$, reading the balance of account $a$ returns $q$ and $p$ decides the value of $R[q]$.
\end{proof}

To prove that \kasset\ has consensus number at most $k$, we reduce \kasset\ to \kcons.
A \kcons\ object exports a single operation \emph{propose} that, the first $k$ times it is invoked, returns the argument of the first invocation.
All subsequent invocations return $\bot$.
Given that \kcons\ is known to have consensus number exactly
$k$~\cite{JT92}, a wait-free algorithm implementing \kasset\ using
only registers and \kcons\ objects implies that the consensus number of \kasset\ is not more than $k$.

\begin{figure}
\hrule \vspace{1mm}
 {\small
\setcounter{linenumber}{0}
\begin{tabbing}
bbb\=bb\=bb\=bb\=bb\=bb\=bb\=bb \=  \kill
Shared variables:\\
\> $AS$, atomic snapshot object\\
\> for each $a\in\A$:\\
\>\> $R_a[i], i \in \Pi$, registers, initially $[\bot, \ldots,\bot]$\\
\>\> $kC_a[i], i \geq 0$, list of instances of \kcons\ objects\\
\\
Local variables:\\
\>\> $\textit{hist}$: a set of completed trasfers, initially empty\\
\>for each $a\in \A$:\\
\>\> \textit{committed}$_a$, initially $\emptyset$\\
\>\> \textit{round}$_a$, initially $0$\\
\\
Upon \textit{transfer}$(a, b, x)$:\\
\nnll\> \textbf{if} $p \notin \mu(a)$ \textbf{then}\\
\nnll\>\> \textbf{return} {\false}\\ 
\nnll\> $tx = (a, b, x, p,\textit{round}_a)$\\
\nnll\label{alg:kwaitfree:register}\> $R_{a}[p].write(tx)$\\
\nnll\> $\textit{collected} = \textit{collect}(a) \setminus \textit{committed}_{a}$\\
\nnll\> \textbf{while} $tx \in \textit{collected}$ \textbf{do} \\
\nnll\label{alg:kwaitfree:picknext}\>\> $\textit{req} = $ the oldest transfer in $\textit{collected}$\\
\nnll\label{alg:kwaitfree:snapshot}\>\> $\textit{prop} = \textit{proposal}(\textit{req}, AS.snapshot())$\\
\nnll\label{alg:kwaitfree:decide}\>\> $\textit{decision} = kC_{a}[\textit{round}_a].propose(prop)$\\
\nnll\>\> $\textit{hist} = \textit{hist} \cup \{\textit{decision}\}$\\
\nnll\label{alg:kwaitfree:update}\>\> $AS.\textit{update}(\textit{hist})$\\
\nnll\>\> $\textit{committed}_a = \textit{committed}_a \cup \{t : \textit{decision}=(t,*)\}$\\
\nnll\>\> $\textit{collected} = \textit{collected} \setminus \textit{committed}_a$\\
\nnll\>\> $\textit{round}_a = \textit{round}_a + 1$\\
\nnll\label{alg:kwaitfree:return}\> \textbf{if} $(tx, \texttt{success}) \in \textit{hist}$ \textbf{then}\\
\nnll\label{alg:kwaitfree:response-true}\>\> \textbf{return} $\true$\\
\nnll\> \textbf{else}\\
\nnll\label{alg:kwaitfree:response-false}\>\> \textbf{return} $\false$\\
\\
Upon \textit{read}$(a)$:\\
\nnll\label{alg:kwaitfree:read-snapshot}\> \textbf{return} \textit{balance}$(a, AS.snapshot())$\\
\\
\textit{collect(a)}:\\
\nnll\> $\textit{collected} = \emptyset$\\
\nnll\> \textbf{for all} $i = \Pi$ \textbf{do}\\
\nnll\>\> \textbf{if} $R_a[i].read() \neq \bot$ \textbf{then}\\
\nnll\>\>\> $\textit{collected} = \textit{collected} \cup \{R_a[i].read()\}$\\
\nnll\> \textbf{return} $collected$\\
\\
\textit{proposal}$((a,b,q,x), \textit{snapshot})$:\\
\nnll\> \textbf{if} $\textit{balance}(a, \textit{snapshot}) \geq x$ \textbf{then}\\
\nnll\>\> $\textit{prop} =  ((a, b, q, x), \texttt{success})$\\
\nnll\> \textbf{else}\\
\nnll\>\> $\textit{prop} =  ((a, b, q, x), \texttt{failure})$\\
\nnll\> \textbf{return} $prop$\\
\\
\textit{balance(a, snapshot)}:\\
\nnll\> $\textit{incoming} = \{tx:  tx = (*, a, *,*,*) \wedge (tx, \texttt{success}) \in \textit{snapshot}\}$\\
\nnll\> $\textit{outgoing} = \{tx:  tx= (a, *, *,*,*) \wedge (tx, \texttt{success}) \in \textit{snapshot}\}$\\
\nnll\> \textbf{return} $q_0(a) + \left(\sum_{(*, a, x,*,*) \in \textit{incoming}} x\right) - \left(\sum_{(a, *, x,*,*) \in \textit{outgoing}} x\right)$
\end{tabbing}
 }
 \hrule
\caption{Wait-free implementation of a {\kasset} object using
  \kcons\ objects. Code for process $p$.}
\label{fig:kaccount-to-kcons}
\end{figure}

The algorithm reducing $k$-shared \textsf{asset-transfer} to \kcons\ is given in~\cref{fig:kaccount-to-kcons}.
Before presenting a formal correctness argument, we first informally explain the intuition of the algorithm.
In our reduction, we associate a series of \kcons\
objects with every account $a$.
Up to $k$ owners of $a$ use the \kcons\
objects to agree on the order of outgoing transfers for $a$.

We maintain the state of the implemented \kasset\ object is maintained using an
atomic snapshot object $AS$.
Every process $p$ uses a distinct entry of $AS$ to store a set $\textit{hist}$. $\textit{hist}$ is a subset of all completed outgoing transfers from accounts that $p$ owns (and thus is allowed to debit).
For example, if $p$ is the owner of accounts $d$ and $e$, $p$'s $\textit{hist}$ contains outgoing transfers from $d$ and $e$.
Each element in the $\textit{hist}$ set is represented as $((a, b, x, s, r), \textit{result})$,
where $a, b$, and $x$ are the respective source account, destination account, and the amount transferred, $s$ is the originator of the transfer,
and $r$ is the \emph{round} in which the transfer  was invoked by the originator.
The value of $\textit{result} \in \{\texttt{success}, \texttt{failure}\}$ indicates whether the transfer succeeds or fails.
A transfer becomes ``visible'' when any process inserts it in its corresponding entry of $AS$.

To read the balance of account $a$, a process takes a snapshot of
$AS$, and then sums the initial balance $q_0(a)$ and amounts of all successful incoming transfers,
and subtracts the amounts of successful outgoing transfers found in $AS$.
We say that a successful transfer $tx$ is in a snapshot $AS$ (denoted by $(tx, \texttt{success}) \in AS$) if there exists an entry $e$ in $AS$ such that $(tx, \texttt{success}) \in AS[e]$.

To execute a transfer $o$ outgoing from account $a$, a process $p$ first
announces $o$ in a register $R_a$ that can be written by $p$ and read
by any other process.
This enables a ``helping'' mechanism needed to ensure wait-freedom
to the owners of $a$~\cite{Her91}.

Next, $p$ collects the transfers proposed by other owners and
tries to agree on the order of the collected transfers and their
results using a series of {\kcons} objects.

A transfer-result pair as a proposal for the next
instance of \kcons\ is chosen as follows.
Process $p$ picks the ``oldest'' collected but not yet committed
operation (based on the round number $\textit{round}_a$ attached to
the transfer operation when a process announces it; ties are broken using process IDs).   
Then $p$ takes a snapshot of $AS$ and checks whether account $a$ has sufficient balance according to the state represented by the snapshot, and equips the transfer with a corresponding \texttt{success} / \texttt{failure} flag.
The resulting transfer-result pair constitutes $p$'s proposal for the next instance of \kcons.
%
\ignore{
For each round, $p$ appends to its copy of $a$ the decided-upon list of transfer-result pairs for that round.
$p$ keeps executing rounds of \kcons\ until $tx$ has been included in the decision of some round.
This can happen either by $p$'s proposal being decided in some round
or by other processes including $tx$ in their proposals and those proposals being decided (possibly over the course of several rounds).
$p$ locally keeps track of all transfers that have been agreed upon so far and excludes them from the set of collected transfers before constructing its proposal in each round.

Note that the transfer $tx$ announced by $p$ will eventually be included in some decision.
If $p$'s proposal is eventually decided, $tx$ will be decided as part of $p$'s proposal (by the algorithm, $p$ always collects its own transfer).
If $p$ executes alone, $p$'s proposal will eventually be decided in some round, namely the first round for which no other process has invoked the \kcons\ object yet.
The only case where $p$'s proposal is never decided is when at least one other process $q$ successfully invokes \kcons\ infinitely many times, with $q$'s proposal being decided.
In this case, however, $tx$ will eventually also be collected by $q$, and thus become part of $q$'s proposal.
Thus, $p$'s \textit{transfer} operation will always eventually terminate.
}
The currently executed transfer by process $p$ returns as soon as
it is decided by a \kcons\ object, the flag of the decided value (\textit{success}/\textit{failure})
indicating the transfer's response (\true/\false).

\begin{lemma}
  \label{lem:at-to-cons}
  The \kasset\ object type has a wait-free implementation in the
  read-write shared memory model equipped with $k$-consensus objects.
\end{lemma}
\begin{proof}
  We essentially follow the footpath of the proof of Theorem~\ref{th:waitfree}. 
  Fix an execution $E$ of the algorithm in
  Figure~\ref{fig:kaccount-to-kcons}.
Let $H$ be the history of $E$. 

To perform a transfer $o$ on an account $a$, $p$ \emph{registers} it in $R_a[p]$
(line~\ref{alg:kwaitfree:register}) and then proceeds through a series
of {\kcons} objects, each time collecting $R_{a}$ to learn about the
transfers concurrently proposed by other owners of $a$.
Recall that each {\kcons} object is wait-free. 
Suppose, by contradiction, that $o$ is registered in $R_{a}$ but is
never decided by any  instance of {\kcons}.
Eventually, however, $o$ becomes the request with the lowest round
number in $R_{a}$ and, thus, some instance of {\kcons} will be
only accessed with $o$ as a proposed value (line~\ref{alg:kwaitfree:decide}).
By validity of {\kcons}, this instance will return $o$ and, thus, $p$ will be able to complete $o$.       

Let $\textit{Ops}$ be the set of all complete operations and
all {\transfer} operations $o$ such that some process
completed the update operation  (line~\ref{alg:kwaitfree:update}) in
$E$ with an argument including $o$
(the atomic snapshot and {\kcons}
operation has been linearized).
Intuitively, we include in $\textit{Ops}$ all operations that
\emph{took effect}, either by returning a response to the user or by
affecting other operations.
Recall that every such {\transfer} operation was 
agreed upon in an instance of {\kcons}, let it be  $kC^{o}$.
Therefore, for every such {\transfer} operation $o$, we can identify the
process $q^{o}$ whose proposal has been decided in that instance. 
We now determine a completion of $H$ and, for each $o\in\textit{Ops}$, we define a linearization point as
follows:

\begin{itemize}

  \item If $o$ is a {\read} operation, it linearizes at the
    linearization point of the snapshot operation (line~\ref{alg:kwaitfree:read-snapshot}).

 \item   If $o$ is a {\transfer} operation that returns {\false},
    it linearizes at the linearization point of the snapshot operation (line~\ref{alg:kwaitfree:snapshot})
    performed by $q^{o}$ just before it invoked
   $kC^{o}.\textit{propose}()$.
    
 \item   If $o$ is a {\transfer} operation that some process
included in the update operation (line~\ref{alg:kwaitfree:update}), it linearizes at the
    linearization point of the \emph{first} update operation in $H$
    (line~\ref{alg:kwaitfree:update}) that includes~$o$.
 Furthermore, if $o$ is incomplete in $H$, we complete it with response $\true$.

\end{itemize}  

Let $\bar H$ be the resulting complete history and let $L$ be the sequence
of complete operations of $\bar H$ in the order of their
linearization points in $E$.
Note that, by the way we linearize operations, the linearization of a
prefix of $E$ is a prefix of $L$.
Also, by construction, the linearization point of an operation
belongs to its interval.  

Now we show that $L$ is legal and, thus, $H$ is
linearizable.
We proceed by induction, starting with the empty (trivially legal)
prefix of $L$.
Let $L_{\ell}$ be the legal prefix of the first $\ell$ operation and 
$op$ be the $(\ell+1)$st operation of $L$.
Let $op$ be invoked by process $p$.
The following cases are possible:

\begin{itemize}

\item $op$ is a {\read}$(a)$: the snapshot taken at $op$'s linearization point
  contains all successful transfers concerning $a$ in $L_{\ell}$. By
  the induction hypothesis, the resulting balance is non-negative.  

\item $op$ is a failed {\transfer}$(a,b,x)$: the snapshot taken at the linearization point of $op$
  contains all successful transfers concerning $a$ in $L_{\ell}$. By
  the induction hypothesis, the balance corresponding to this snapshot
   non-negative. By the algorithm, the balance is  less than $x$. 
 
\item  $op$ is a successful  {\transfer}$(a,b,x)$.
  Let $L_{s}$, $s\leq\ell$, be the prefix of $L_{\ell}$ that only
  contains operations linearized before the moment of time when $q^{o}$
  has taken the snapshot just before accessing $kC^{o}$.

  As before accessing $kC^{o}$, $q$ went through all preceding
  {\kcons} objects associated with $a$ and put the decided values in
  $AS$,  $L_{s}$ must include \emph{all}  outgoing {\transfer}
  operations for $a$. Furthermore, $L_s$ includes a \emph{subset} of
  all incoming transfers on $a$.
  Thus, $\textit{balance}(a,L_k) \leq \textit{balance}(a,L_{\ell})$.
  
  By the algorithm, as $op={\transfer}(a,b,x)$ succeeds, we have  
  $\textit{balance}(a,L_k)\geq x$.
  Thus, $\textit{balance}(a,L_{\ell})\geq x$ and the resulting balance
  in $L_{\ell+1}$ is non-negative.

\end{itemize}    

Thus, $H$ is linearizable. 
\end{proof}  


\begin{theorem}
  A $k$-shared \textsf{asset-transfer} object has consensus number $k$.
\end{theorem}

\begin{proof}
  It follows directly from \cref{lem:cons-to-at} that $k$-shared \textsf{asset-transfer} has consensus number at least $k$.
  Moreover, it follows from \cref{lem:at-to-cons} that $k$-shared \textsf{asset-transfer} has consensus number at most $k$.
  Thus, the consensus number of $k$-shared \textsf{asset-transfer} is exactly $k$.
\end{proof}


\section{Related Work}
\label{sec:related-work}

In this section we discuss related work with regards to various parts of \name.
We start by considering other asset transfer systems, both for the private (permissioned) and public (permissionless) setting.
We then elaborate on ordering of inputs, and finally we discuss work related to secure broadcast, which is an important building block of \name.

\subsection{Asset Transfer Systems}
\label{sec:sys-related}
The different flavors of \name\ are suitable for both a private (permissioned) and public (permissionless) setting.
A private setting implies the assumption of an access control mechanism, specifying who is allowed to participate in the system.
In this case, we assume that this mechanism is external to the system itself.
Private protocols, such as Corda~\cite{he16corda}, Hyperledger Fabric~\cite{hyperledger}, or Vegvisir~\cite{kar18vegvisir} rely on such a mechanism.

Importantly, the access control mechanism rules out the possibility of Sibyl attacks~\cite{do02sybil}, where a malicious party can take control over a system by using many identities, toppling the one third assumption on the fraction of Byzantine participants.
Once this is done, the malicious party can engage in a double-spending attack.

Decentralized systems for the public, i.e., permissionless, setting are open to the world.
They do not have an explicit access control mechanism and allow anyone to join.
Systems which fall into this category include Bitcoin~\cite{nakamotobitcoin}, Ethereum \cite{ethereum}, Avalanche~\cite{rocket}, ByzCoin~\cite{kogi16byzcoin}, Algorand~\cite{gilad2017algorand}, Hybrid consensus~\cite{pa17hybrid}, PeerCensus~\cite{decker2016bitcoin}, or Solida~\cite{abr16solida}.
To prevent malicious parties from overtaking the system, these systems rely on Sybil-proof techniques, e.g., proof-of-work~\cite{nakamotobitcoin}, or proof-of-stake~\cite{be16cryptocurrencies}.

These systems, whether they address the permissionless or the permissioned environment, seek to solve consensus in their implementation.
It is worth noting that many of these solutions can allow for more than just transfers, and enable access to smart contracts.
Our focus is, however, on decentralized transfer systems, and the surprising result of this paper is that we can implement such a system without resorting to consensus.


To deploy \name\ in a public setting where anybody can participate, Sibyl attacks need to be addressed and the system needs to be scalable to accommodate many participants.
This is possible by using a Sybil-resistant and scalable implementation of secure broadcast, while keeping the transfer algorithm (see~\Cref{fig:banking-relaxed}) unchanged.
To the best of our knowledge, there is no such implementation of secure broadcast in the literature so far.
Our probabilistic secure broadcast implementation is scalable; transfers can be accepted within $O(\log(N))$ message delays
ensuring security with overwhelming probability.
We also show how to make our protocol Sybil-resistant (see \cref{section:permissionless}) and thus deployable in an open (permissionless) environment.



\subsection{Ordering Constraints}

In the blockchain ecosystem, there exist several efforts to avoid building a totally ordered chain of transfers.
The idea is to replace the totally ordered linear structure of a blockchain with that of a directed acyclic graph (DAG) for structuring the transfers in the system.
Notable systems in this spirit include Byteball~\cite{chu6byteball}, Vegvisir~\cite{kar18vegvisir}, the GHOST protocol~\cite{som13accelerating}, Corda~\cite{he16corda}, or Nano~\cite{lemahieu2018nano}.
Even if these systems use a DAG to replace the classic blockchain, their algorithms still employ consensus.
As we show here,  total order (obtained via consensus) is not necessary for implementing decentralized transfers.

We can also use a DAG to characterize the relation between transfers in \name, but we do not resort to solving consensus to build the DAG, nor do we use the DAG in order to solve consensus.
More precisely, we can regard each account as having an individual history.
Each such history is managed by the corresponding account owner without depending on a global view of the system state.
Another way to characterize \name{} is that we do not build a totally ordered sequence of all transfers in the system (like a classic blockchain~\cite{nakamotobitcoin}).
Instead we maintain a per-owner sequence of transfers.
Each sequence is ordered individually (by its respective owner), and is loosely coupled with other sequences (through dependencies established by causality).

Another parallel we can draw is between the accounts in our system (as we defined them in~\Cref{sec:def-sequential}) and conflict-free replicated data types (CRDTs)~\cite{sh11crdt}.
Specifically, similar to a CRDT, in \name{} we support concurrent updates on different accounts while preserving consistency.
Since each account has a unique owner, this rules out the possibility of conflicting operations on each (correct) account.
In turn, this ensures that the state at correct nodes always converges to a consistent version.
In the terminology of \cite{sh11crdt}, we provide strong eventual consistency.


As we explained earlier, the ordering among transfers in \name is based on causality, as defined through the happened-before relationship of Lamport~\cite{lam78time}.
Various causally-consistent algorithms exist~\cite{cac11intro,llo11cops}.
One problematic aspect in these algorithms is that the metadata associated with tracking dependencies can be a burden~\cite{ak2016cure,bail12potential,me17causal}.
This happens because such algorithms track all \emph{potential} causal dependencies.
In our \name algorithm for the message passing model (\Cref{fig:banking-relaxed}) we track dependencies \emph{explicitly}~\cite{bail12potential}, permitting a more efficient implementation with a smaller set of dependencies.
More concretely, we specify that each transfer outgoing from an account only depends on previous transfers outgoing from and incoming to that---and only that---account, ignoring the transfers that affect other (irrelevant) accounts.


\subsection{Asynchronous Agreement Protocols}

The important insight that an asynchronous broadcast-style abstraction suffices for transfers appears in the literature as early as 2002, due to Pedone and Schiper~\cite{ped02handling}.
Duan et. al.~\cite{du18beat} introduce efficient Byzantine fault-tolerant protocols for storage and also build on this insight, as does recent work by Gupta~\cite{gup16nonconsensus} on financial transfers which is the closest to us.
To the best of our knowledge, however, we are the first to consider this insight formally, prove it, and build practical systems around consensusless algorithms.

An important strength of the whole \name class of algorithms is that they are asynchronous, meaning that they do not have to rely inherently on timeouts to ensure liveness guarantees.
This is in contrast to deterministic consensus-based solutions.
Such solutions have to resort to fine-tuned timeout parameters which affect their performance~\cite{bess14state,mil16honeybadger}.

Asynchronous protocols for consensus exist, but they typically employ heavy cryptography relying on randomization to overcome the FLP impossibility~\cite{FLP85,kin10scalable}.
A few recent efforts are trying to make these protocols more efficient~\cite{ali18communication,baz18clairvo,du18beat,mil16honeybadger}.
These asynchronous protocol are designed for a relatively small- to medium-scale, similar to our deterministic algorithm \namedet.
Compared to \namedet, asynchronous consensus protocols have higher complexity and ensure probabilistic guarantees.

\section{Conclusions}
\label{sec:conclusions}

\cmt{YAP}{add complexity comparison (for deterministic and probabilistic approach) wrt to \# mst or \# bits exchanged per transaction}

In this paper we revisited the problem of implementing a decentralized asset transfer system.
Since the rise of the Bitcoin cryptocurrency, this problem has garnered significant innovation.
Most of the innovation, however, has focused on improving the original solution which Bitcoin proposed, namely, that of using a consensus mechanism to build a blockchain where transactions across the whole system are totally ordered.

We did not aim to investigate current consensus-based solutions and push the envelope on performance or other metrics.
Instead, we showed that we can implement a decentralized transfer system without resorting to consensus.
To this end, we first precisely defined the transfer system object type and proved that it has consensus number \emph{one} if a single account is not shared by multiple processes.
That is, it occupies the lowest rank of the consensus hierarchy and can be implemented without the need for solving consensus.
We proved this surprising result by borrowing from the theory of concurrent objects, and then leveraged the result to build \name{}, the first consensusless transfer system.

A consensusless solution for decentralized transfers has multiple advantages.
Concretely, the transfer algorithm in \name{} is not subject to the FLP impossibility~\cite{FLP85}.
Additionally, \name{} is significantly simpler than consensus-based solutions (because it relies on a simpler secure broadcast primitive) and exhibits higher performance.

\cmt{YAP}{this and the following sentences are not understandable without more information, e.g., explaining 100 replicas vs system size..
and why BFT-smart and not another recent consensus algorithm designed for virtual currencies has been used for the comparison.}
We compared the performance of \name{} to that of a transfer system based on BFT-Smart, a state-of-the-art consensus-based state machine replication system.
\name\ provides performance superior to that of the consensus-based solution.
In systems of up to $100$ replicas, regardless of system size, we observed a throughput improvement ranging from $1.5x$ to $6x$, while achieving a decrease in latency of up to $2x$.

We implemented \namedet\ for a private (permissioned) environment that is prone to Byzantine failures and where participants do not need to trust each other.
This implementation is easily extensible to the large-scale permissionless setting, and this is the focus of our concurrent work.
\cmt{YAP}{sub-second transfer execution with probabilistic algo?\\
needs to be compared to other recent implementations of ledger latencies.. (elastico: 800s, omniledger: 14s, rapidchain: 8.5s)\\
For throughput, other approaches have around 3000-8000 tx/s\\
(Expected time to failure up to 4500 years)}
Our preliminary results are highly encouraging.
Most notably, we can obtain sub-second transfer execution on a global scale deployment of thousands of nodes.

\clearpage
\section*{Acknowledgements}
We are thankful to colleagues in the Distributed Computing Lab, as well as members from other groups at EPFL, for many insightful discussions that helped us understand and develop the ideas in this paper.
Discussions with Yvonne-Anne Pignolet have been instrumental in clarifying the limitations of our approach and working out the relation between smart contracts and asset transfer. 
Alex Auvolat-Bernstein discovered a flaw in our \namedet algorithm and helped us rectify it.
We also thank Rida Bazzi for pointing out recent work very close to \namedet .

\bibliographystyle{acm}
\bibliography{refs}

\newpage
\begin{appendices}

\section{Secure broadcast}
\label{app:secure-broadcast}

A central building block of our \namemp algorithm is the secure broadcast primitive.
We now give a more formal discussion of this primitive, which we use for establishing the correctness of \namemp .

Intuitively, secure broadcast ensures that correct processes receive the same sequence of messages broadcast from a given sender, even if that sender is malicious (Byzantine).
There are multiple algorithms for secure broadcast,
exhibiting various tradeoffs (\Cref{sec:sb}).
In the rest of this discussion, we draw directly from the pioneering work of Malkhi and Reiter~\cite{MR97srm} (called secure reliable multicast therein).

\begin{sloppy}
Formally, let $\M$ be a set of \emph{messages}.
A secure broadcast protocol exports a $\textit{broadcast}(m)$ primitive, where $m\in\M$, and triggers callback events $\textit{deliver}(p,m)$, where $m\in \M$ and $p\in \Pi$.
The protocol satisfies the following properties:
\end{sloppy}

\begin{itemize}
\item \textbf{Integrity:} a correct process executes $\textit{deliver}(p,m)$
  at most once, and, in case the sender process $p$ is benign, only if $p$ called $\textit{broadcast}(m)$.

\item \textbf{Agreement:} if $p$ and $q$ are correct and $p$ executes $\textit{deliver}(r,m)$, then $q$ eventually executes $\textit{deliver}(r,m)$.

\item \textbf{Validity:} if $p$ is correct and executes $\textit{broadcast}(m)$, then $p$ eventually executes $\textit{deliver}(p,m)$.

\item \textbf{Source Order:} if $p$ and $q$ are benign and $p$ executes
  $\textit{deliver}(r,m)$ before $\textit{deliver}(r,m')$, then
  $q$ does not execute $\textit{deliver}(r,m')$ before executing
  $\textit{deliver}(r,m)$. Moreover, if $r$ is benign and
  broadcasts $m$ and afterwards broadcasts $m'$, then no benign process delivers these two
  messages in the opposite order.

\end{itemize}

Note that, unlike Byzantine agreement~\cite{la82byzantine}, secure broadcast
does not guarantee that \emph{some} message is delivered when the sender
is faulty.


Intuitively, we need to make sure that a message is delivered if and
only if a \emph{quorum} of more than two thirds of processes
\emph{accepted} it.
Every delivered message must be equipped with a \emph{proof}: a set of acknowledgements cryptographically
signed by a quorum of processes proving that they accepted that
message.

To ensure the source order property, upon broadcasting any message, a benign process $p$
tags the message with a monotonically growing sequence number.
A benign process accepts (and acknowledges) a new message received from $p$ only once it has
accepted all messages from $p$ with lower sequence numbers and it received no
different message from $p$ with the same sequence number.
To broadcast a message, the sender sends it to all the processes,
waits until a quorum of them accepts it, and resends the message attaching the corresponding proof.
Once a benign process receives a message equipped with a proof from $p$, it resends
the message to all (to ensure the agreement property) and, once the previous
message broadcast by $p$ has been delivered,
delivers the message.

\section{Correctness of \namemp Algorithm}
\label{app:proof}

We focus on \namemp , our asset transfer algorithm for the message passing setting based on secure broadcast, which we described earlier in~\Cref{fig:banking-relaxed}.
We first prove an auxiliary lemma.

\ignore{
In the following, let $x^p$ denote the value of a local variable $x$ at a correct process $p$.

\begin{lemma}
  \label{lem:prefix}
At any point in the execution,  for all correct processes $p$ and $r$
and for every process $q$, $hist[q]^p$ and $hist[q]^r$ are related by
containment. 
\end{lemma}
\begin{proof}
A correct process updates local variable $hist[q]$ each time a new
transfer from or to $q$
is delivered by secure broadcast and validated in
line~\ref{line:append-outgoing} or~\ref{line:append-incoming}.
By the source order property of secure broadcast (see~\Cref{app:secure-broadcast}), correct processes
$p$ and $r$ deliver messages from $q$ in the same order.
By the algorithm in~\Cref{fig:banking-relaxed}, a message from $q$ with a sequence number $i$ is
added to
$\textit{toValidate}$ set
only if the previous message added to $\textit{toValidate}$ had
sequence number $i-1$ (line~\ref{line:nextrec}).
Similarly, such a message is successfully validated only
if the last validated message from $q$ had sequence number $i-1$
(line~\ref{line:check-hist1}).
Thus, $hist[q]^p$ and $hist[q]^r$ are constructed from the same sequence of
messages from process $q$ and are therefore related by containment.
\end{proof}
}

\begin{lemma}
  \label{lem:liveness}
In any infinite execution of algorithm {\namemp} (Figure~\ref{fig:banking-relaxed}),
every operation performed by a correct process eventually completes.
\end{lemma}
\begin{proof}
A transfer operation that fails or a read operation invoked by a correct process
returns immediately (lines~\ref{line:response-false}
and~\ref{line:response-read}, respectively).

Consider a transfer operation $T$ invoked by a correct process $p$
that \emph{succeeds} (i.e., passes
the check in line~\ref{line:check-balance}), so $p$ broadcasts a message with the transfer details using secure broadcast (line~\ref{line:check-sbcast}).
%
By the validity property of secure broadcast,
$p$ eventually delivers the message (via the secure broadcast callback, line~\ref{line:deliver}) and adds it to the
$\textit{toValidate}$ set.
%
By the algorithm, this message includes a set $\textit{deps}$ of
operations (called $h$, line~\ref{line:sb-deliver}) that involve $p$'s account.
This set includes transfers that process $p$ delivered and validated after issuing
the prior successful outgoing transfer (or since the initial system
time if there is no such transfer)
but before issuing $T$ (lines~\ref{line:check-sbcast} and~\ref{line:dep-null}).

As process $p$ is correct, it operates on its own account, respects the
sequence numbers, and issues a transfer only if it has enough balance
on the account. Thus, when it is delivered by $p$, $T$ must satisfy the first three conditions of the
$\textit{Valid}$ predicate (lines~\ref{line:validation}--\ref{line:check-hist2}).
Moreover, by the algorithm, all dependencies (labeled $h$ in function $\textit{Valid}$) included in $T$  are in the history
$hist[p]$ and, thus the fourth validation condition (line~\ref{line:check-hist3})
also holds.

Thus, $p$ eventually validates $T$ and completes the operation by
returning $\true$ in line~\ref{line:response-true}.
\end{proof}

\begin{theorem}
Algorithm {\namemp} (Figure~\ref{fig:banking-relaxed}) implements an
\textsf{asset-transfer} object type.
\end{theorem}
\begin{proof}
Fix an execution $E$ of the algorithm, let $H$ be the
  corresponding  history.

Let $\V$ denote the set of all messages 
  that were delivered (line~\ref{line:deliver}) and
validated (line~\ref{line:validation}) at correct processes in $E$.
Every message $m=[(q,d,y,s),h]\in\V$ is put
in $hist[q]$ (line~\ref{line:append-outgoing}) and $hist[d]$ (line~\ref{line:append-incoming}).
%
We define an order $\preceq \subseteq \V\times\V$ as
follows. For $m=[(q,d,y,s),h]\in \V$ and $m'=[(r,d',y',s'),h']\in \V$,
we have $m\preceq m'$ if and only if one of the following conditions holds:

\begin{itemize}
\item $q=r$ and $s<s'$,

\item $(r,d',y',s')\in h$, or

\item there exists $m''\in \V$ such that $m\preceq m''$ and
  $m''\preceq m'$.
\end{itemize}

By the source order property of secure broadcast (see~\Cref{app:secure-broadcast}), correct processes
$p$ and $r$ deliver messages from any process $q$ in the same order.
By the algorithm in~\Cref{fig:banking-relaxed}, a message from $q$ with a sequence number $i$ is
added by a correct process to
$\textit{toValidate}$ set
only if the previous message from $q$ added to $\textit{toValidate}$ had
sequence number $i-1$ (line~\ref{line:nextrec}).
Furthermore, a message  $m=[(q,d,y,s),h]$ is validated at a correct
process only if all messages in $h$ have been previously validated (line~\ref{line:check-hist3}).
Therefore, $\preceq$ is acyclic and
thus can be extended to a total order.

Let $S$ be the sequential history constructed from any such total order on messages
in $\V$ in which every message $m=[(q,d,y,s),h]$ is replaced with the
invocation-response pair
$\textit{transfer}(q,d,y);\true$.

By construction, every operation $\textit{transfer}(q,d,y)$ in $S$ is preceded by a sequence of
transfers that ensure that the balance of $q$ does not drop below $y$
(line~\ref{line:check-hist2}).
In particular, $S$ includes all outgoing transfers from the account of
$q$ performed previously by $q$ itself.
Additionally $S$ may order some \emph{incoming} transfer to $q$
that did not appear at $hist[q]$ before the corresponding $(q,d,y,s)$
has been added to it.
But these ``unaccounted'' operations may only increase the balance
of $q$ (line~\ref{line:append-incoming}) and, thus, it is indeed legal
to return $\true$.

By construction, for each correct process $p$, $S$ respects the order of
successful transfers issued by $p$.
Thus, the subsequence of successful transfers in $H$ ``looks'' linearizable to
the correct processes:  $H$, restricted to successful transfers
witnessed by the correct processes, is consistent with a
a legal sequential history $S$.

Let $p$ be a correct process in $E$.
Now let $\V_p$ denote the set  of all messages 
  that were delivered (line~\ref{line:deliver}) and
validated (line~\ref{line:validation}) at $p$  in $E$.
Let $\preceq_p$ be the subset of $\preceq$ restricted to the elements in
$\V_p$.
Obviously, $\preceq_p$ is cycle-free and we can again extend it to a
total order.
Let $S_p$ be the sequential history build in the same way as $S$
above.
Similarly, we can see that $S_p$ is legal and, by construction,
consistent with the local history of \emph{all} operations of $p$
(including reads and failed transfers).

By Lemma~\ref{lem:liveness}, every operation invoked by a correct
process eventually completes.
Thus, $E$ indeed satisfies the properties of an \textsf{asset-transfer} object type.
\end{proof}

\section{Analysis of \erg}
\label{appendix:erganalysis}

We now discuss the correctness of \erg .

\paragraph{No duplication}: A correct process maintains a $delivered$ variable that it checks and updates before delivering a message. This prevents any correct process from delivering more than one message.
\paragraph{Integrity}: Before broadcasting a message, the sender signs that message with its private key. Before delivering a message $m$, a correct process verifies $m$'s signature. Under the assumption that signatures cannot be forged, this prevents any correct process from delivering a message that was not previously broadcast by the sender.
\paragraph{Validity}: Upon broadcasting a message, the sender also immediately delivers it. Since this happens \emph{deterministically}, and thus \erg\ satisfies $0$-validity, independently from the parameter $G$.

\subsection{Totality}

\erg\ satisfies $\epsilon_t$-totality with $\epsilon_t$ upper-bounded by a function that decays exponentially with $G$, and polynomially increases with the fraction $f$ of Byzantine faults.

Indeed, the network of connections established among the correct processes is an undirected Erdős–Rényi graph, and totality is satisfied if such graph is connected. This allows us to bound the probability of totality not being satisfied, using a well-known result on the connectivity of Erdős–Rényi graphs.

Upon initialization, a correct process randomly selects a sample of other processes (it uses an oracle to achieve this, see Assumption \ref{ergpdeassumption:sampling}) with which it will exchange messages.

\begin{figure}
\begin{tikzpicture}
\begin{semilogyaxis}[
    title={Gossip validity ($N = 1024$)},
    xlabel={Sample size ($G$)},
    ylabel={Failure probability ($\epsilon$)},
    legend pos=south west,
    ymajorgrids=true,
    grid style=dashed,
    width=0.45\textwidth
]

\addplot[
    color=palette-0,
    mark=square,
    ]
    coordinates {
    (5, 0.0468971)
    (6, 0.00616465)
    (7, 0.00082308)
    (8, 0.000109936)
    (9, 1.46593e-05)
    (10, 1.95095e-06)
    (11, 2.59128e-07)
    (12, 3.43493e-08)
    (13, 4.54415e-09)
    (14, 5.99956e-10)
    (15, 7.90524e-11)
    (16, 1.03953e-11)
    (17, 1.36422e-12)
    (18, 1.78672e-13)
    };

\addplot[
    color=palette-1,
    mark=square,
    ]
    coordinates {
    (5, 0.0750802)
    (6, 0.0108075)
    (7, 0.00159513)
    (8, 0.000235931)
    (9, 3.48508e-05)
    (10, 5.13878e-06)
    (11, 7.56299e-07)
    (12, 1.11098e-07)
    (13, 1.6289e-08)
    (14, 2.38374e-09)
    (15, 3.48173e-10)
    (16, 5.07579e-11)
    (17, 7.38554e-12)
    (18, 1.07258e-12)
    (19, 1.55468e-13)
    };

\addplot[
    color=palette-2,
    mark=square,
    ]
    coordinates {
    (6, 0.018723)
    (7, 0.00304488)
    (8, 0.000497592)
    (9, 8.12597e-05)
    (10, 1.32489e-05)
    (11, 2.15635e-06)
    (12, 3.50334e-07)
    (13, 5.68154e-08)
    (14, 9.19744e-09)
    (15, 1.48623e-09)
    (16, 2.39728e-10)
    (17, 3.8598e-11)
    (18, 6.20332e-12)
    (19, 9.9516e-13)
    (20, 1.59357e-13)
    };

\addplot[
    color=palette-3,
    mark=square,
    ]
    coordinates {
    (6, 0.0324293)
    (7, 0.00579829)
    (8, 0.00104639)
    (9, 0.000188894)
    (10, 3.4054e-05)
    (11, 6.12932e-06)
    (12, 1.10135e-06)
    (13, 1.97562e-07)
    (14, 3.53787e-08)
    (15, 6.32472e-09)
    (16, 1.12875e-09)
    (17, 2.01101e-10)
    (18, 3.57672e-11)
    (19, 6.35053e-12)
    (20, 1.12561e-12)
    (21, 1.99166e-13)
    };

\addplot[
    color=palette-4,
    mark=square,
    ]
    coordinates {
    (6, 0.0562503)
    (7, 0.0110174)
    (8, 0.00219357)
    (9, 0.000437621)
    (10, 8.72306e-05)
    (11, 1.73625e-05)
    (12, 3.45044e-06)
    (13, 6.84614e-07)
    (14, 1.3562e-07)
    (15, 2.68227e-08)
    (16, 5.29646e-09)
    (17, 1.04416e-09)
    (18, 2.05519e-10)
    (19, 4.03861e-11)
    (20, 7.92335e-12)
    (21, 1.55196e-12)
    (22, 3.03492e-13)
    };

    \legend{$f = 0.00$, $f = 0.05$, $f = 0.10$, $f = 0.15$, $f = 0.20$}

\end{semilogyaxis}
\end{tikzpicture}
\begin{tikzpicture}
\begin{semilogyaxis}[
    title={Gossip validity ($N = 8192$)},
    xlabel={Sample size ($G$)},
    ylabel={Failure probability ($\epsilon$)},
    legend pos=south west,
    ymajorgrids=true,
    grid style=dashed,
    width=0.45\textwidth
]

\addplot[
    color=palette-0,
    mark=square,
    ]
    coordinates {
    (6, 0.0514704)
    (7, 0.00680596)
    (8, 0.00091692)
    (9, 0.000123815)
    (10, 1.67209e-05)
    (11, 2.25766e-06)
    (12, 3.04757e-07)
    (13, 4.11286e-08)
    (14, 5.54916e-09)
    (15, 7.48522e-10)
    (16, 1.00943e-10)
    (17, 1.36094e-11)
    (18, 1.83441e-12)
    (19, 2.472e-13)
    };

\addplot[
    color=palette-1,
    mark=square,
    ]
    coordinates {
    (6, 0.0908409)
    (7, 0.0130731)
    (8, 0.00194178)
    (9, 0.000289716)
    (10, 4.32464e-05)
    (11, 6.45465e-06)
    (12, 9.63164e-07)
    (13, 1.43691e-07)
    (14, 2.14316e-08)
    (15, 3.19581e-09)
    (16, 4.76438e-10)
    (17, 7.10118e-11)
    (18, 1.05816e-11)
    (19, 1.57643e-12)
    (20, 2.34798e-13)
    };

\addplot[
    color=palette-2,
    mark=square,
    ]
    coordinates {
    (7, 0.0251123)
    (8, 0.00410308)
    (9, 0.000676103)
    (10, 0.000111543)
    (11, 1.84027e-05)
    (12, 3.03557e-06)
    (13, 5.00617e-07)
    (14, 8.25422e-08)
    (15, 1.36066e-08)
    (16, 2.24249e-09)
    (17, 3.695e-10)
    (18, 6.08698e-11)
    (19, 1.00252e-11)
    (20, 1.65078e-12)
    (21, 2.71762e-13)
    };

\addplot[
    color=palette-3,
    mark=square,
    ]
    coordinates {
    (7, 0.0482708)
    (8, 0.00863774)
    (9, 0.0015701)
    (10, 0.000286165)
    (11, 5.21728e-05)
    (12, 9.51094e-06)
    (13, 1.73349e-06)
    (14, 3.15885e-07)
    (15, 5.75501e-08)
    (16, 1.04827e-08)
    (17, 1.90902e-09)
    (18, 3.47581e-10)
    (19, 6.32721e-11)
    (20, 1.15154e-11)
    (21, 2.09533e-12)
    (22, 3.81186e-13)
    };

\addplot[
    color=palette-4,
    mark=square,
    ]
    coordinates {
    (7, 0.0935578)
    (8, 0.0181991)
    (9, 0.0036432)
    (10, 0.000733401)
    (11, 0.000147782)
    (12, 2.97794e-05)
    (13, 5.99996e-06)
    (14, 1.20865e-06)
    (15, 2.43425e-07)
    (16, 4.9017e-08)
    (17, 9.86832e-09)
    (18, 1.98634e-09)
    (19, 3.99743e-10)
    (20, 8.04306e-11)
    (21, 1.618e-11)
    (22, 3.25423e-12)
    (23, 6.54384e-13)
    };

    \legend{$f = 0.00$, $f = 0.05$, $f = 0.10$, $f = 0.15$, $f = 0.20$}

\end{semilogyaxis}
\end{tikzpicture}

\begin{tikzpicture}
\begin{semilogyaxis}[
    title={Gossip validity ($N = 65536$)},
    xlabel={Sample size ($G$)},
    ylabel={Failure probability ($\epsilon$)},
    legend pos=south west,
    ymajorgrids=true,
    grid style=dashed,
    width=0.45\textwidth
]

\addplot[
    color=palette-0,
    mark=square,
    ]
    coordinates {
    (7, 0.0559772)
    (8, 0.00739694)
    (9, 0.00099765)
    (10, 0.000134924)
    (11, 1.82536e-05)
    (12, 2.46955e-06)
    (13, 3.34099e-07)
    (14, 4.51982e-08)
    (15, 6.11439e-09)
    (16, 8.27126e-10)
    (17, 1.11886e-10)
    (18, 1.51346e-11)
    (19, 2.04715e-12)
    (20, 2.76895e-13)
    };

\addplot[
    color=palette-1,
    mark=square,
    ]
    coordinates {
    (8, 0.0157051)
    (9, 0.00233295)
    (10, 0.000348508)
    (11, 5.21041e-05)
    (12, 7.79065e-06)
    (13, 1.16485e-06)
    (14, 1.74164e-07)
    (15, 2.60394e-08)
    (16, 3.89307e-09)
    (17, 5.82024e-10)
    (18, 8.70116e-11)
    (19, 1.30077e-11)
    (20, 1.94451e-12)
    (21, 2.90676e-13)
    };

\addplot[
    color=palette-2,
    mark=square,
    ]
    coordinates {
    (8, 0.0334044)
    (9, 0.00544519)
    (10, 0.000897848)
    (11, 0.000148321)
    (12, 2.45091e-05)
    (13, 4.05007e-06)
    (14, 6.69252e-07)
    (15, 1.10587e-07)
    (16, 1.82729e-08)
    (17, 3.01926e-09)
    (18, 4.98861e-10)
    (19, 8.24228e-11)
    (20, 1.36177e-11)
    (21, 2.24982e-12)
    (22, 3.71688e-13)
    };

\addplot[
    color=palette-3,
    mark=square,
    ]
    coordinates {
    (8, 0.0715139)
    (9, 0.012696)
    (10, 0.00230695)
    (11, 0.00042095)
    (12, 7.68682e-05)
    (13, 1.40382e-05)
    (14, 2.56377e-06)
    (15, 4.68204e-07)
    (16, 8.5503e-08)
    (17, 1.56141e-08)
    (18, 2.85128e-09)
    (19, 5.20657e-10)
    (20, 9.50719e-11)
    (21, 1.73597e-11)
    (22, 3.16972e-12)
    (23, 5.78745e-13)
    };

\addplot[
    color=palette-4,
    mark=square,
    ]
    coordinates {
    (9, 0.0296394)
    (10, 0.00591346)
    (11, 0.00119086)
    (12, 0.000240262)
    (13, 4.84912e-05)
    (14, 9.78732e-06)
    (15, 1.97543e-06)
    (16, 3.98703e-07)
    (17, 8.04686e-08)
    (18, 1.62403e-08)
    (19, 3.27756e-09)
    (20, 6.61449e-10)
    (21, 1.33485e-10)
    (22, 2.69374e-11)
    (23, 5.43588e-12)
    (24, 1.09692e-12)
    (25, 2.21343e-13)
    };

    \legend{$f = 0.00$, $f = 0.05$, $f = 0.10$, $f = 0.15$, $f = 0.20$}

\end{semilogyaxis}
\end{tikzpicture}
\begin{tikzpicture}
\begin{semilogyaxis}[
    title={Gossip validity ($N = 524288$)},
    xlabel={Sample size ($G$)},
    ylabel={Failure probability ($\epsilon$)},
    legend pos=south west,
    ymajorgrids=true,
    grid style=dashed,
    width=0.45\textwidth
]

\addplot[
    color=palette-0,
    mark=square,
    ]
    coordinates {
    (8, 0.0607705)
    (9, 0.00801589)
    (10, 0.00108106)
    (11, 0.000146232)
    (12, 1.97882e-05)
    (13, 2.67791e-06)
    (14, 3.62398e-07)
    (15, 4.90427e-08)
    (16, 6.63683e-09)
    (17, 8.98145e-10)
    (18, 1.21543e-10)
    (19, 1.6448e-11)
    (20, 2.22583e-12)
    (21, 3.01211e-13)
    };

\addplot[
    color=palette-1,
    mark=square,
    ]
    coordinates {
    (9, 0.0188311)
    (10, 0.00279416)
    (11, 0.00041741)
    (12, 6.24181e-05)
    (13, 9.33518e-06)
    (14, 1.39618e-06)
    (15, 2.08816e-07)
    (16, 3.12307e-08)
    (17, 4.67088e-09)
    (18, 6.98577e-10)
    (19, 1.04479e-10)
    (20, 1.56258e-11)
    (21, 2.33696e-12)
    (22, 3.49511e-13)
    };

\addplot[
    color=palette-2,
    mark=square,
    ]
    coordinates {
    (9, 0.0444297)
    (10, 0.00721138)
    (11, 0.00118842)
    (12, 0.000196341)
    (13, 3.2451e-05)
    (14, 5.36382e-06)
    (15, 8.86591e-07)
    (16, 1.46545e-07)
    (17, 2.42225e-08)
    (18, 4.00373e-09)
    (19, 6.61774e-10)
    (20, 1.09384e-10)
    (21, 1.80798e-11)
    (22, 2.98837e-12)
    (23, 4.93938e-13)
    };

\addplot[
    color=palette-3,
    mark=square,
    ]
    coordinates {
    (10, 0.0186189)
    (11, 0.00337568)
    (12, 0.000615813)
    (13, 0.000112467)
    (14, 2.05441e-05)
    (15, 3.75289e-06)
    (16, 6.8556e-07)
    (17, 1.25235e-07)
    (18, 2.28772e-08)
    (19, 4.17906e-09)
    (20, 7.63403e-10)
    (21, 1.39453e-10)
    (22, 2.54741e-11)
    (23, 4.6534e-12)
    (24, 8.5004e-13)
    };

\addplot[
    color=palette-4,
    mark=square,
    ]
    coordinates {
    (10, 0.0483275)
    (11, 0.00957399)
    (12, 0.00192555)
    (13, 0.000388451)
    (14, 7.8412e-05)
    (15, 1.583e-05)
    (16, 3.19587e-06)
    (17, 6.45205e-07)
    (18, 1.30258e-07)
    (19, 2.62974e-08)
    (20, 5.30906e-09)
    (21, 1.07182e-09)
    (22, 2.16383e-10)
    (23, 4.36843e-11)
    (24, 8.81912e-12)
    (25, 1.78043e-12)
    (26, 3.59436e-13)
    };

    \legend{$f = 0.00$, $f = 0.05$, $f = 0.10$, $f = 0.15$, $f = 0.20$}

\end{semilogyaxis}
\end{tikzpicture}
\caption{Upper bound for the $\epsilon$-security of \erg, as a function of the gossip sample size ($G$), the fraction of Byzantine failures ($f$), and the system size ($N$).
}
\label{figure:ergsecurity}
\end{figure}
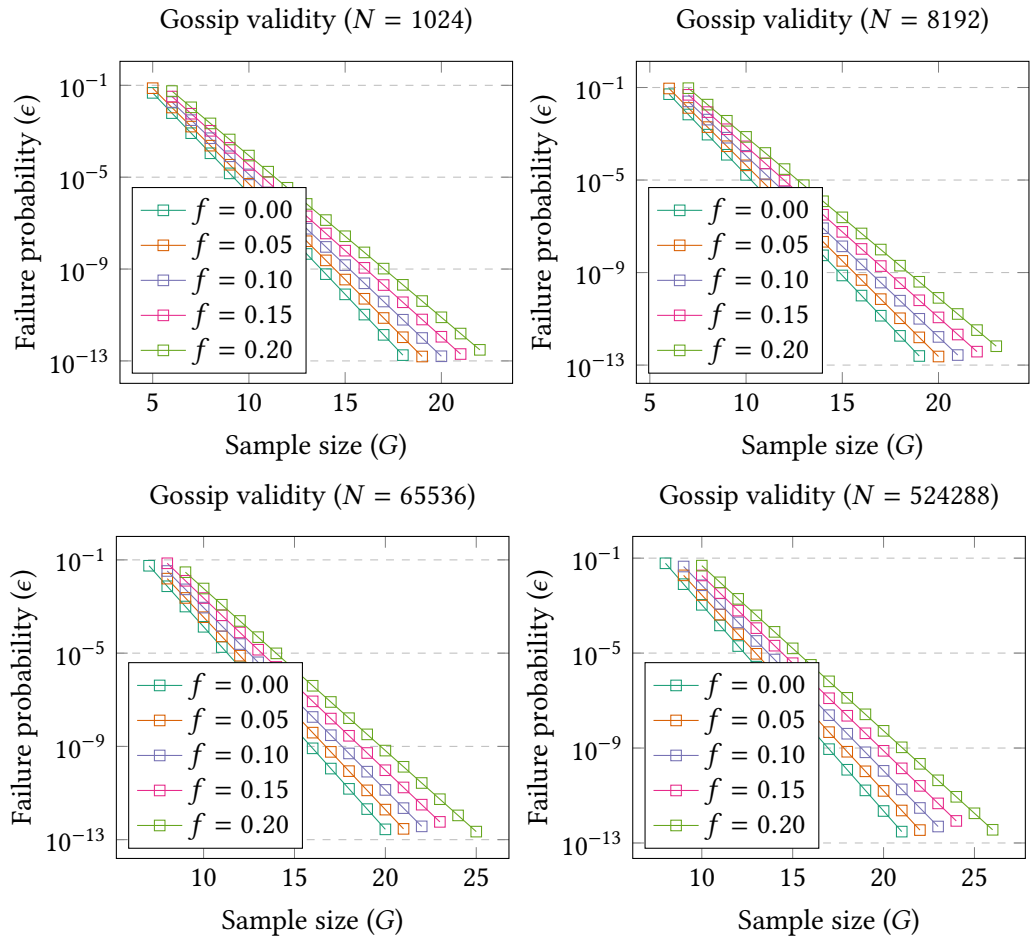

We start by noting that every link is eventually reciprocated by correct processes, i.e., if a correct process $\pi$ is in the sample of $\rho$, then $\rho$ will eventually be in the sample of $\pi$ (this is due to the fact that messages are always eventually delivered, see Assumption \ref{ergpdeassumption:asynchrony}).

We consider the sub-graph of connections only between \emph{correct} processes. This network is eventually undirected. We show that, if such graph is connected, then \erg\ satisfies totality. This is due to the fact that every message will eventually propagate through all the gossip links, reaching every correct process (again, due to Assumption \ref{ergpdeassumption:asynchrony}).

We show that any two correct processes have an independent probability of being connected. This is due to the fact that, upon initialization, the number of elements in a correct process' gossip sample is sampled from a Poisson distribution. Poisson distributions quickly limit to binomial distributions for large systems, and we show that selecting a binomially distributed number of distinct objects from a set is equivalent to selecting each object with an independent probability. This proves that the sub-graph of connections between correct processes is an Erdős–Rényi graph.

Erdős–Rényi graphs are well known in literature \cite{phasetransitions} to display a connectivity phase transition: when the expected number of connections each node has exceeds the logarithm of the number of nodes, the probability of the graph being connected steeply increases from $0$ to $1$ (in the limit of infinitely large systems, this increase becomes a step function). We use that result to compute the probability of the sub-graph of correct processes being connected and, consequently, of \erg\ satisfying totality.

\Cref{figure:ergsecurity} shows the $\epsilon$-security of \erg, as a function of the gossip sample size ($G$), the fraction of Byzantine failures ($f$) and the size of the system ($N$).
\cmt{YAP}{why only to f=0.2?\\
to make failure probability more concrete, give an example what 10-30 means, e.g., expected time to failure assuming x transactions per second}
\cmt{YAP}{asymptotic size of sample as a function of n}


\section{Analysis of \pde}
\label{appendix:pdeanalysis}

We now discuss the correctness of \pde .

\paragraph{No duplication}
A correct process maintains a \emph{delivered} variable that it checks and updates before delivering a message. This prevents any correct process from delivering more than one message.

\paragraph{Integrity}: Before broadcasting a message, the sender signs that message with its private key. Before delivering a message $m$, a correct process verifies $m$'s signature. Under the assumption that signatures cannot be forged, this prevents any correct process from delivering a message that was not previously broadcast by the sender.

In order to study validity, totality and consistency, we first establish some auxiliary results.

\subsection{Auxiliary results}
\label{subsection:pdeauxiliary}

For a correct process, the execution of \pde\ reduces to three operations: publishing an {\tt Echo} message, publishing a {\tt Ready} message, and delivering a message. Each operation is triggered by one or more conditions:
\begin{itemize}
    \item A correct process publishes an {\tt Echo} message upon \broadcastinstance.Deliver of a {\tt Send} message.
    \item A correct process publishes a {\tt Ready} message upon collecting either enough {\tt Echo} messages from its echo sample, or enough {\tt Ready} messages from its ready sample.
    \item A correct process delivers a message upon collecting enough {\tt Ready} messages from its delivery sample.
\end{itemize}

We study the probability of each condition being fulfilled in steps.

\begin{definition}[Ready, E-ready, R-ready]
Let $\pi$ be a correct process, let $m$ be a message. Then:
\begin{itemize}
    \item $\pi$ is \textbf{Ready} for $m$ if it eventually publishes a {\tt Ready} message for $m$.
    \item $\pi$ is \textbf{E-ready} for $m$ if it is ready for $m$ as a result of having collected enough {\tt Echo} messages for $m$.
    \item $\pi$ is \textbf{R-ready} for $m$ if it is ready for $m$ as a result of having collected enough {\tt Ready} messages for $m$.
\end{itemize}
\end{definition}

Let $\pi$ be a correct process, let $m$ be a message.

\paragraph{E-ready probability} We compute lower and upper bounds for the probability of $\pi$ being E-ready for $m$, given the number of correct processes that echo (i.e., publish an {\tt Echo} message for) $m$.

Echo samples are uniformly picked with replacement from the set of processes (correct processes use an oracle to achieve this, see Assumption \ref{ergpdeassumption:sampling}). Therefore, each element of $\pi$'s echo sample has an independent probability of being Byzantine and, if correct, of having echoed $m$.

Let $\rho$ be an element of $\pi$'s echo sample. We consider two scenarios. In the first scenario, no Byzantine process ever sends an {\tt Echo} message for $m$. In this scenario, the probability of $\pi$ receiving an {\tt Echo} message for $m$ from $\rho$ reduces to the probability of $\rho$ being correct and having echoed $m$ (this is due to the fact that messages are always eventually delivered, see Assumption \ref{ergpdeassumption:asynchrony}).

In the second scenario, all Byzantine processes send an {\tt Echo} message for $m$ to all the correct subscribed processes. In this scenario, the probability of $\pi$ receiving an {\tt Eccho} message for $m$ from $\rho$ reduces to the probability of $\rho$ being Byzantine, or $\rho$ being correct and having echoed $m$.

The probability of $\pi$ being E-ready for $m$ is minimized in the first scenario, and maximized in the second. Both probabilities can be computed by noting that the number of {\tt Echo} messages that $\pi$ receives is binomially distributed.

\paragraph{Ready feedback} We study the feedback mechanism produced by correct processes being R-ready for $m$, i.e., ready as a result of having received enough {\tt Ready} messages for $m$.

Ready samples are uniformly picked with replacement from the set of processes. We define a random \textbf{ready multigraph} $g$ allowing multi-edges and loops whose nodes represent the set of correct processes. The predecessors\footnote{Node $\mu$ is a predecessor of node $\nu$ in the multigraph $g$ if $\rp{\mu \rightarrow \nu} \in g$. Since multigraphs allow for multiple edges, the set of predecessors of a node is a multiset.} of a node $\nu$ represent the correct processes in $\nu$'s ready sample.

We introduce \pebbling, a game played on the nodes of $g$ where readiness for $m$ spreads like a disease: an \textbf{infected} node represents a correct process that is ready for $m$, and whenever enough predecessors of a node $\nu$ are infected, $\nu$ becomes infected (i.e., R-ready) as well.

\pebbling\ is played in rounds. At the beginning of each round, a \textbf{player} infects an arbitrary set of nodes (this models a set of processes being E-ready for $m$). Throughout the rest of the round, the disease automatically propagates to all the nodes that have enough infected predecessors.

Given the number of rounds and the number of nodes infected per round, we use Markov chains to compute the probability distribution underlying the number of nodes that are infected at the end of \pebbling.

We use \pebbling\ in three settings:
\begin{itemize}
    \item When evaluating \textbf{validity}, a correct sender takes the role of the player in a single-round game of \pebbling.

    When broadcasting, the correct sender uses \broadcast\ to distribute $m$. Unless the totality of \broadcast\ is compromised, every correct process eventually publishes an {\tt Echo} message for $m$. Given the fraction of Byzantine processes, we compute the probability of any correct process being E-ready for $m$.

    The number of correct processes that are E-ready for $m$ is used as input to (the single round of) \pebbling; its outcome represents the number of correct processes that are ready for $m$. This allows us to compute the probability of a correct process (and, in particular, the sender) eventually delivering $m$.

    \item When evaluating \textbf{totality}, a Byzantine sender takes the role of the player in a multi-round game of \pebbling. The game starts with no infected nodes. At the beginning of each round, the player infects one healthy node; throughout the round, the infection propagates until either all nodes are infected, or no uninfected node has enough infected predecessors.

    In order to compromise totality, a Byzantine adversary must cause at least one, but not all correct processes to deliver a message.

    We show that a Byzantine adversary that can cause processes to become E-ready for arbitrary messages can easily perform an attack to compromise totality. We therefore restrict ourselves to the case where the Byzantine adversary can arbitrarily cause any correct process to be E-ready for only one message $m$, and assume that totality is compromised if the adversary can cause two distinct processes to be E-ready for two distinct messages. By doing so, we effectively compute an upper bound for the probability of totality being compromised.

    Under Assumption \ref{ergpdeassumption:anonymity}, the adversary has no knowledge of any correct process' ready sample. As a result of this, we later show (see \cref{subsection:pdetotality}) that every action implemented by the adversary can be modeled by a multi-round game of \pebbling.

    Intuitively, lacking any information to meaningfully distinguish correct processes with respect to the topology of the \emph{ready multigraph}, the minimal action the adversary can perform on the system is causing a single correct process to become E-ready for $m$. As a result, zero or more additional correct processes become R-ready. The adversary can then carry on, causing more correct processes to become E-ready for $m$, or stop.

    In a game of \pebbling, therefore, the adversary's strategy reduces to playing in rounds, and stopping as soon as any process delivers $m$. At the end of each round, we compute the probability of at least one, but not all correct processes delivering $m$. Totality is compromised is this happens in any round.

    \item When evaluating \textbf{consistency}, a Byzantine sender takes the role of the player in a single-round game of \pebbling.

    Let $m$ and $m'$ be two conflicting messages. As we do when evaluating totality, we compute an upper bound for consistency by assuming that if any two correct processes become E-ready for $m$ and $m'$, then consistency is compromised.

    Under the assumption that correct processes can be E-ready for at most one message, if $m$ and $m'$ are delivered by at least one correct process, then either $m$ or $m'$ is delivered without any correct process being E-ready for it.

    In order to compute the probability of $m$ being delivered even if no correct process is E-ready for $m$, we consider a scenario where all Byzantine processes send a {\tt Ready} message for $m$ to their correct subscribers. Noting how, in this scenario, the behavior of a Byzantine process is identical to that of a correct process that is E-ready for $m$, we can compute the probability of $m$ being delivered by playing a game of \pebbling\ where Byzantine processes are included as initially infected nodes.
\end{itemize}


\paragraph{Delivery probability} We compute lower and upper bounds for the probability of $\pi$ delivering $m$, given the number of correct processes that are ready for $m$.

This can be achieved using the same technique that we employed when computing lower and upper bounds for the probability of a process being E-ready for $m$, given the number of correct process that echoed $m$. For the sake of brevity, we don't repeat that analysis here.

\paragraph{Early consistency} We call \textbf{early consistency} the condition where no two correct processes are E-ready for two different messages, $m_1$ and $m_2$.

As we see in \cref{subsection:pdetotality,subsection:pdeconsistency}, we compute upper bounds for the probability of compromising totality and consistency by assuming that if early consistency is compromised, then both totality and consistency are compromised.

Under Assumption \ref{ergpdeassumption:anonymity}, the adversary has no knowledge of the echo sample of any correct process. Therefore, the adversary has no way of meaningfully distinguishing two correct processes, based on the effect that their {\tt Echo} messages will have on the system.





We consider a scenario where an adversary releases two different messages $m_1$ and $m_2$, and can:
\begin{itemize}
    \item Cause any correct process to echo any of $m_1$ and $m_2$.
    \item Determine if any correct process is E-ready for $m_1$ or $m_2$.
\end{itemize}
Early consistency is compromised if at least one correct process is E-ready for $m_1$, and one correct process is E-ready for $m_2$.

We start by noting that, if the echo threshold is larger than half of the sample size, then the order in which the adversary causes each process to echo either $m_1$ or $m_2$ does not affect the probability of compromising early consistency.

Therefore, any adversary that causes $n_1$ correct processes to echo $m_1$, and $n_2$ correct processes to echo $m_2$, has the same probability of compromising early consistency as one that \emph{first} causes $n_1$ correct processes to echo $m_1$, \emph{then} $n_2$ correct processes to echo $m_2$.

Moreover, the probability of any correct process being E-ready for $m_2$ is an increasing function of $n_2$. Given that at least one correct process is E-ready for $m_1$, the probability of compromising early consistency is maximized by the adversary that maximizes $n_2$.

Therefore, the probability of compromising early consistency is maximized by an adversary that causes one correct process at a time to echo $m_1$, until at least one correct process is E-ready for $m_1$, then causes all the other correct processes to echo $m_2$.

An adversary could release more than two different messages. We argue, however, that the adversary maximizes the probability of violating early consistency using only two messages. If the adversary needs to make one process E-ready for $m_1$ and another process E-ready for $m_2$, it needs enough echoes for $m_1$ as well as enough choes for $m_2$. The more echoes for each respective message, the higher the chance of a process becoming E-ready for it. Intuitively, introducing a third message would only decrease the chance of a correct process becoming E-ready for either of $m_1$ and $m_2$.

\paragraph{Feedback creation} We compute an upper bound for the probability of $m$ being delivered by any correct process, given that no correct process is E-ready for $m$.

In order to do so, we consider a scenario where every Byzantine process sends a {\tt Ready} message for $m$ to every correct subscriber. Noting that, in this scenario, a Byzantine process behaves identically to a correct process that is E-ready for $m$, we compute the probability underlying the number of correct processes that are ready for $m$ by playing a game of \pebbling\ where also the Byzantine processes are included as initially infected nodes.

We then use the upper bound on the delivery probability to bound the probability of $m$ being delivered.

\subsection{Validity}
\label{subsection:pdevalidity}

We compute an upper bound for the probability of \textbf{validity} being compromised.

Validity is compromised if a correct sender broadcasts, but does not deliver, a message $m$. We compute an upper bound for the probability of this happening by assuming that, if the totality of \broadcast\ is compromised, then the validity of \pde\ is compromised as well.

If the totality of \broadcast\ is not compromised, then every correct process publishes an {\tt Echo} message for $m$. We use our lower bound on the E-ready probability to compute the probability distribution underlying the number of correct processes that are E-ready for $m$.

We then play a game of \pebbling\ to compute the probability distribution underlying the number of correct processes that are ready for $m$.

Finally, we use our lower bound for the delivery probability to bound the probability that a correct process (and, specificaly, the sender) will not deliver $m$.

\subsection{Totality}
\label{subsection:pdetotality}

We compute an upper bound for the probability of \textbf{totality} being compromised.

Totality is compromised if at least one, but not all processes deliver a message. We compute an upper bound for the probability of this happening by assuming that totality is compromised if early consistency is compromised.

If early consistency is not compromised, then a Byzantine sender will cause processes to be E-ready for at most one message $m$.

We consider a scenario where a Byzantine adversary can:
\begin{itemize}
    \item Cause any correct process to be E-ready for $m$.
    \item Determine if any correct process delivered $m$.
\end{itemize}

Under Assumption \ref{ergpdeassumption:anonymity}, the adversary has no knowledge of neither the ready nor the delivery sample of any correct process. Therefore, the adversary has no way of meaningfully distinguishing two correct processes, based on the effect that their {\tt Ready} messages will have on the system.

Let $n$ represent the number of correct processes that are E-ready for $m$. The number of processes that deliver $m$ is a non decreasing function of $n$. Therefore, the probability of compromising totality, given that at least one correct process delivered $m$, is maximized by the adversary that minimizes $n$.

The adversary that has the highest probability of compromising totality, therefore, will cause one correct process at a time to be E-ready for $m$ until at least one correct process delivers $m$. Totality is compromised if at least one correct process does not deliver $m$.

We model the above using a multi-round game of \pebbling, where the player infects one more uninfected node at the beginning of each round. At the end of each round, we use both our bounds for the delivery probability to determine whether or not totality can be compromised at that round.

\subsection{Consistency}
\label{subsection:pdeconsistency}

We compute an upper bound for the probability of \textbf{consistency} being compromised.

Consistency is compromised if two correct processes deliver two conflicting messages, $m$ and $m'$. We compute an upper bound for the probability of this happening by assuming that consistency is compromised if early consistency is compromised.

If early consistency is not compromised, but both $m$ and $m'$ are delivered by at least one correct process, then either $m$ or $m'$ is delivered without any correct process being E-ready for it. We bound the probability of this happening using our result on feedback creation.

\section{Analysis of \tfayto}
\label{appendix:tfaytoanalysis}

We now discuss the correctness of \tfayto.

\paragraph{No creation} A correct process delivers a message only if it was previously \singleshotinstance.Delivered. Moreover, \singleshotinstance.Broadcast is invoked by the sender process only upon \multishotinstance.Broadcast. 

Since \singleshot\ satisfies \textbf{no duplication}, each instance of \singleshotinstance\ will deliver at most once. Since \singleshot\ also satisfies \textbf{integrity}, no correct process will \singleshotinstance.Deliver without a corresponding invocation of \singleshotinstance.Broadcast on the side of the sender.
    
Therefore, if the sender never broadcasts more than $n$ messages, then no correct process will deliver more than $n$ messages.

\paragraph{Integrity} A correct process delivers a message only if it was previously \singleshotinstance.Delivered. Moreover, \singleshotinstance.Broadcast is invoked by the sender process only upon \multishotinstance.Broadcast. 
    
Since \singleshot\ satisfies \textbf{integrity}, then no correct process will deliver a message $m$, unless $m$ was previously broadcast by the sender.

\subsection{Multi-validity}

\Multishot\ satisfies $\epsilon$-\textbf{multi-validity} if the underlying abstraction of \singleshot\ satisfies $\epsilon$-\textbf{validity}.

If the sender $\sigma$ is correct, and it initially broadcasts $m_1, \ldots, m_n$, then $\sigma$ \singleshotinstance.$(i - 1)$.Broadcasts $m_i$ for every $i \in \{1, \ldots, m_i\}$

If \singleshot\ satisfies $\epsilon$-validity, then each $m_i$ is eventually \singleshotinstance.Delivered by $\sigma$ with probability at least $(1 - \epsilon)$. Therefore, $m_1, \ldots, m_n$ are eventually \singleshotinstance.Delivered with probability at least $(1 - \epsilon)^n$.

Upon \singleshotinstance.Deliver of $m_i$, $messages[i - 1]$ is set to $m_i$. When all $m_1, \ldots, m_n$ are \singleshotinstance.Delivered, the first $n$ entries of $messages$ are updated to a value different than $\bot$. As a result, $m_1, \ldots, m_n$ are delivered in sequence.

Consequently, the sender process delivers $m_1, \ldots, m_n$ with probability at least $(1 - \epsilon)^n$.

\subsection{Multi-totality}

\Multishot\ satisfies $\epsilon$-\textbf{multi-totality} if the underlying abstraction of \singleshot\ satisfies $\epsilon$-\textbf{totality}.

Let $\pi$ be a correct process. If $\pi$ delivered $n$ messages, then $messages[j] \neq \bot$ for all $j \in \{0, \ldots, n - 1\}$. Moreover, if $\pi$ is correct, then $messages[j]$ is set to a value other than $\bot$ only upon \singleshotinstance.$j$.Deliver.

Each instance of \singleshotinstance\ satisfies totality with a probability at least $(1 - \epsilon)$. Therefore, with probability at least $(1 - \epsilon)^n$, all \singleshotinstance.$j$ satisfy totality, and every correct process eventually sets $messages[j] \neq \bot$, with $j \in \{0, \ldots, n - 1\}$.

Consequently, every correct process delivers $n$ messages with probability at least $(1 - \epsilon)^n$.

\subsection{Multi-consistency}

\Multishot\ satisfies $\epsilon$-\textbf{multi-consistency} if the underlying abstraction of \singleshot\ satisfies $\epsilon$-\textbf{consistency}.

Let $\pi, \rho$ be two correct processes. If $\pi$ and $\rho$ initially delivered $m_1, \ldots, m_n$ and $m_1', \ldots, m_n'$ respectively, then, for every $i \in \{1, \ldots, n\}$, $\pi$ \singleshotinstance.$(i - 1)$.Delivered $m_i$, and $\rho$ \singleshotinstance.$(i - 1)$.Delivered $m_i'$.

Each instance of \singleshot\ satisfies consistency with probability at least $(1 - \epsilon)$. Therefore, with probability at least $(1 - \epsilon)^n$, \singleshotinstance.$j$ satisfies consistency for all $j \in \{0, \ldots, n - 1\}$, and $m_i = m_i'$ for all $i \in \{1, \ldots, n\}$.

Therefore, with probability at least $(1 - \epsilon)^n$, all correct processes that initially deliver $n$ messages deliver $m_1, \ldots, m_n$.
\end{appendices}

\end{document}